\documentclass[conference,letterpaper]{IEEEtran}
\usepackage{amsmath,amssymb,amsthm,mathtools}
\usepackage{booktabs}
\usepackage{cite}
\usepackage{microtype}
\usepackage{url}
\usepackage[hidelinks]{hyperref}
\usepackage{flushend}
\usepackage{xcolor}

\usepackage{tikz}
\usepackage{graphicx}
\usetikzlibrary{arrows.meta,positioning,calc,fit,shapes.misc}

\usepackage{algorithm}
\usepackage{algpseudocode}

\usepackage{amsmath,amssymb}
\usepackage{braket}

\newcommand{\HD}{\operatorname{HD}}
\newcommand{\Dec}{\operatorname{Dec}}
\newcommand{\clip}{\operatorname{clip}}
\newcommand{\Enc}{\operatorname{Enc}}
\newcommand{\Z}{\mathbb{Z}}

\usepackage{algorithm} 
\usepackage{algpseudocode}
\usepackage{amsmath,amssymb,braket}

\newtheorem{theorem}{Theorem}
\newtheorem{lemma}[theorem]{Lemma}
\newtheorem{proposition}[theorem]{Proposition}
\newtheorem{corollary}[theorem]{Corollary}
\newtheorem{definition}{Definition}
\newtheorem{remark}{Remark}
\newtheorem{claim}{Claim}

\DeclareMathOperator{\Tr}{Tr}
\newcommand{\eps}{\varepsilon}
\newcommand{\ketbra}[1]{\lvert #1\rangle\!\langle #1\rvert}

\def\IsDraft{} 

\ifdefined\IsDraft
    \newcommand{\authnote}[2]{{\bf [{\color{red} #1's Note:} {\color{blue} #2}]}}
\else
    \newcommand{\authnote}[2]{}
\fi

\title{Quantum Advantage for Two-Party Differential Privacy}

\author{\IEEEauthorblockN{Daniel Alabi$^1$, Emil T. Khabiboulline$^{2,3}$}
\IEEEauthorblockA{$^1$Electrical and Computer Engineering,
University of Illinois at Urbana-Champaign, Urbana, IL 61801, USA}
\IEEEauthorblockA{$^2$Joint Center for Quantum Information and Computer Science, NIST/University of Maryland, College Park, MD 20742, USA}
\IEEEauthorblockA{$^3$Joint Quantum Institute, NIST/University of Maryland, College Park, MD 20742, USA}
alabid@illinois.edu, ekhabibo@umd.edu}

\hypersetup{
  pdftitle={Quantum Advantage for Two-Party Differential Privacy},
  pdfauthor={Daniel Alabi, Emil T. Khabiboulline},
  pdfsubject={Quantum two-party differential privacy},
  pdfkeywords={differential privacy, quantum communication, Hamming distance, retained workspace, honest-but-curious adversaries}
}

\begin{document}
\maketitle

\begin{abstract}

We introduce information-theoretically private quantum protocols for two-party Hamming distance when both parties must output the same estimate. Classically, for input length $n$, information-theoretic protocols require $\Omega(\sqrt{n})$ error under pure differential privacy and $\Omega(\sqrt{n}/\log n)$ error under strong approximate differential privacy, whereas computational security permits $O(1)$ error. In Klauck's honest, nonpreemptive, message-preserving model, we give an $O(n)$-communication quantum protocol with pure $\varepsilon$ quantum differential privacy (QDP) and expected error at most $\frac{2}{\sinh \varepsilon}+\gamma$, for every $\gamma>0$. For approximate $(\varepsilon, \delta)$ QDP, an exact hockey-stick divergence calculation yields strictly smaller error, while preserving the $O(1)$-versus-$\Omega(\sqrt{n}/\log n)$ separation for $\delta=o(1/n)$. Thus, quantum communication achieves $O(1)$ information-theoretic error, matching the accuracy available classically only under computational assumptions.

The main construction uses a guarded coherent round trip and an equal-Gram rigidity principle that prevents an honest player from retaining input-dependent complementary information. We separate this model from weaker prescribed-channel privacy, which already admits an exact classical realization, and from fully retention-robust security, against which measurement-and-abort attacks remain possible. Therefore, we identify preservation of non-orthogonal quantum messages as a resource for privacy.

\end{abstract}


\section{Introduction}

Differential privacy (DP) is best understood in the trusted
curator setting: a
curator collects the data, evaluates a low-sensitivity statistic, and releases
a noisy answer~\cite{dwork2006calibrating,DworkKMMN06}. For Hamming distance, a function with sensitivity of one, this gives
expected error \(O(1/\varepsilon)\), independent of the input length \(n\)~\cite{McGregorMPRTV10}.
In many applications, however, no trusted curator exists~\cite{BNO08,Ben-SassonBHR19,KL14}. Instead, two
mutually distrustful parties hold separate inputs and wish to compute a common
statistic~\cite{kushilevitz1997communication,rao2020communication} while protecting each input from the other party. Requiring privacy
simultaneously against both participants fundamentally changes the problem:
the transcript needed to coordinate the computation may itself reveal
information that a trusted curator would never release~\cite{dwork2004privacy,DSSU17}.

McGregor, Mironov, Pitassi, Reingold, Talwar, and Vadhan initiated the
systematic study of lower bounds in this two-party setting~\cite{McGregorMPRTV10, mcgregor2011limits}. For Hamming distance, they showed
that classical information-theoretic privacy is dramatically less accurate
than central differential privacy: for suitable constant privacy parameters,
purely differentially private protocols require
\(\Omega(\sqrt{n})\) error, and strong approximate privacy requires
\(\Omega(\sqrt{n}/\log n)\) error~\cite{McGregorMPRTV10, mcgregor2011limits}. By contrast, secure function evaluation can
emulate the trusted curator and achieve \(O(1)\) error when privacy is required
only against computationally bounded parties~\cite{Haitner22,Haitner25}.

This led McGregor et al.~\cite{McGregorMPRTV10, mcgregor2011limits} to pose the following question:
\begin{quote}
\emph{What is the largest gap in accuracy between
information-theoretic and
computationally differentially private two-party protocols?}
\end{quote}
They showed that a linear separation is possible for certain functions.
For natural functions such as Hamming distance, however, their results expose
a particularly important gap:
\[
\underbrace{O(1)}_{\substack{\text{central or computational}\\
                             \text{differential privacy}}}
\qquad\text{versus}\qquad
\underbrace{\widetilde{\Omega}(\sqrt{n})}_{\substack{\text{classical
information-theoretic}\\ \text{two-party privacy}}}.
\]
This raises a basic question left outside the classical theory:
\textit{is the loss in accuracy inherent to information-theoretic privacy itself, or
does it arise from the classical, copyable nature of communication?}

The distinction matters both conceptually and cryptographically. If the
classical lower bound extended to quantum protocols,
then to obtain the accuracy of the central model in the
two-party differential privacy would inherently require
computational assumptions or
trusted setup~\cite{MPRV09}. Conversely, an information-theoretic quantum protocol with
\(O(1)\) error would show that computational assumptions are not the only way
to cross the classical accuracy barrier~\cite{Haitner22,Haitner25,BNO08}. It would identify noncopyability~\cite{Quantum_copying_96} and
coherent information flow as genuine resources for
achieving differential privacy, not
merely for reducing
communication complexity or performing cryptographic key distribution~\cite{Klauck07,Bennett-1984-Qc,Vazirani-2014-Fd}.
Moreover, determining exactly which adversarial models permit such an
advantage is essential: a separation obtained only by discarding prescribed
workspace has a different meaning from one that remains secure when a party
may retain arbitrary quantum information~\cite{Vidick_Wehner_2023}.

In this work, we show that the answer is model dependent. Alice holds
\(x\in\{0,1\}^{n}\), Bob holds \(y\in\{0,1\}^{n}\), and both parties must
produce the same random estimate \(\widehat d\) of
\[
d=\HD(x,y)=\bigl|\{i:x_i\neq y_i\}\bigr|.
\]
Under Klauck's honest, nonpreemptive, message-preserving definition~\cite{klauck2004quantum}, we give
the first information-theoretic quantum protocol that crosses the classical
two-party accuracy barrier. For every \(\gamma>0\), the protocol is pure
\(\varepsilon\)-quantum differentially private, communicates \(O(n)\) qubits
and bits, and satisfies
\[
\mathbb{E}\!\left[|\widehat d-d|\right]
<
\frac{2}{\sinh\varepsilon}+\gamma.
\]
Thus, for constant \(\varepsilon\), quantum communication achieves \(O(1)\)
error, whereas every classical information-theoretic protocol under the same
honesty model requires \(\Omega(\sqrt n)\) error. In this model, quantum
communication therefore (information-theoretically)
recovers the accuracy that
is available classically only through computationally secure evaluation~\cite{NN01}.

We also obtain a strict improvement under approximate differential privacy. Rather than
using the pure-DP noise parameter \(\varepsilon\), we compute the privacy using the hockey-stick divergence and select the minimum
\(\alpha_n^\star(\varepsilon,\delta)>\varepsilon\) for a particular
\(\delta\)-budget~\cite{hirche2023quantum}. The resulting protocol is
\((\varepsilon,\delta)\)-quantum differentially private and has expected error
at most
\[
\frac{2}{\sinh\alpha_n^\star(\varepsilon,\delta)}+\gamma
<
\frac{2}{\sinh\varepsilon}+\gamma.
\]
For \(\delta=o(1/n)\), this gives \(O(1)\) quantum error compared with the
classical lower bound \(\Omega(\sqrt n/\log n)\).

\begin{figure}[t]
\centering
\resizebox{\columnwidth}{!}{%
\begin{tikzpicture}[
    >=Latex,
    font=\small,
    party/.style={
        draw,
        rounded corners=6pt,
        minimum width=2.0cm,
        minimum height=2.9cm,
        align=center,
        thick
    },
    smallbox/.style={
        draw,
        rounded corners=5pt,
        align=center,
        inner sep=4pt,
        thick
    },
    dashedbox/.style={
        draw=orange!80!black,
        dashed,
        rounded corners=5pt,
        align=center,
        inner sep=5pt,
        thick
    },
    bigcallout/.style={
        draw,
        rounded corners=6pt,
        align=center,
        inner sep=6pt,
        thick
    }
]

\node[font=\bfseries\large] at (0,4.8)
{Classical two-party DP};

\node[party] (alice) at (-2.9,1.8) {\textbf{Alice}\\[0.25cm]
input\\[0.2cm]
$ x \in \{0,1\}^n $};

\node[party] (bob) at (2.9,1.8) {\textbf{Bob}\\[0.25cm]
input\\[0.2cm]
$ y \in \{0,1\}^n $};

\node at (0,2.75) {$m_1,m_2,\ldots,m_r$};
\node at (0,2.15) {$\vdots$};

\draw[->, thick, orange!85!black] ($(alice.east)+(0.1,0.55)$) -- ($(bob.west)+(-0.1,0.55)$);
\draw[<-, thick, orange!85!black] ($(alice.east)+(0.1,-0.05)$) -- ($(bob.west)+(-0.1,-0.05)$);

\node[dashedbox] (transcript) at (0,0.65) {classical transcript $T$};

\node[draw, minimum width=0.45cm, minimum height=0.6cm, thick] (copyicon) at (0,-0.75) {};
\draw[thick] ($(copyicon.center)+(-0.08,-0.12)$) rectangle ($(copyicon.center)+(0.08,0.12)$);
\draw[thick] ($(copyicon.center)+(0.08,0.12)$) -- ($(copyicon.center)+(0.0,0.12)$) -- ($(copyicon.center)+(0.0,0.2)$) -- ($(copyicon.center)+(0.08,0.12)$);

\draw[dashed, thick, ->] (transcript.south) -- (copyicon.north);

\node[align=center] at (1.7,-0.7) {\emph{privacy proofs}\\[-1mm]\emph{reduce to}\\[-1mm]\emph{transcript views}};

\node at (0,-1.7) {copyable transcript};
\draw[dashed, thick, ->] (-0.8,-1.45) -- (-0.16,-0.9);
\draw[dashed, thick, ->] (0.8,-1.45) -- (0.16,-0.9);

\node[smallbox, minimum width=3.2cm] (out) at (0,-3.15) {common output $\hat d$};

\draw[dashed, thick, ->] ($(alice.south)+(0,-0.05)$) |- (out.west);
\draw[dashed, thick, ->] ($(bob.south)+(0,-0.05)$) |- (out.east);

\node[bigcallout, minimum width=7.2cm, minimum height=2.2cm] at (0,-5.7) {%
For Hamming distance,\\
classical information-theoretic $2$-party DP\\
has large error~\cite{McGregorMPRTV10,mcgregor2011limits}:\\[1mm]
pure DP: \quad $\Omega(\sqrt{n})$,\\[1mm]
strong (i.e., $\delta=o(1/n)$) approximate DP: \quad $\Omega\!\left(\dfrac{\sqrt{n}}{\log n}\right)$};

\end{tikzpicture}%
}
\caption{Classical two-party differential privacy. Classical privacy proofs reduce to transcript views. The interaction transcript can be copied.}
\label{fig:classical-2party-dp}
\end{figure}
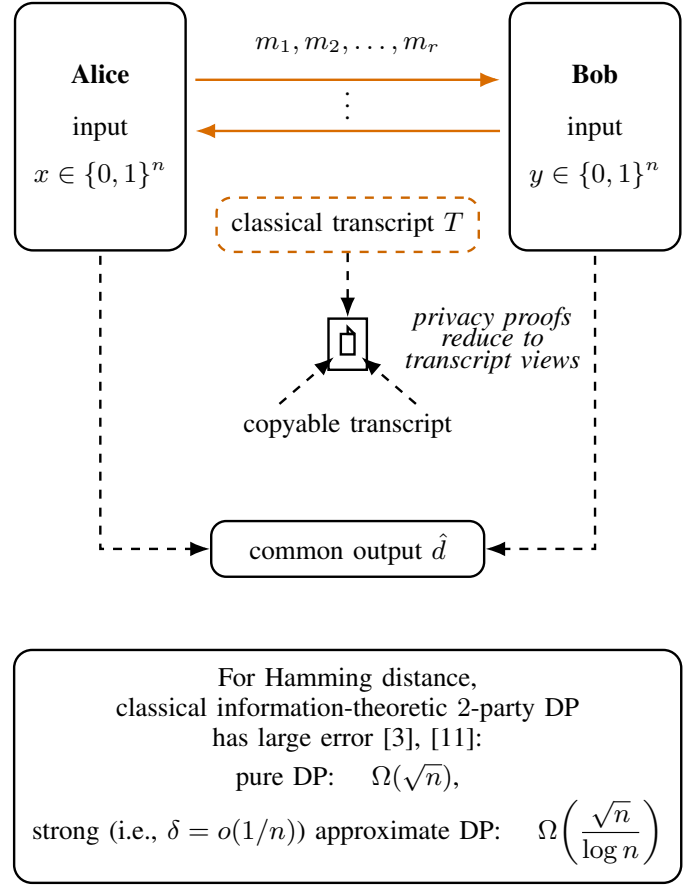

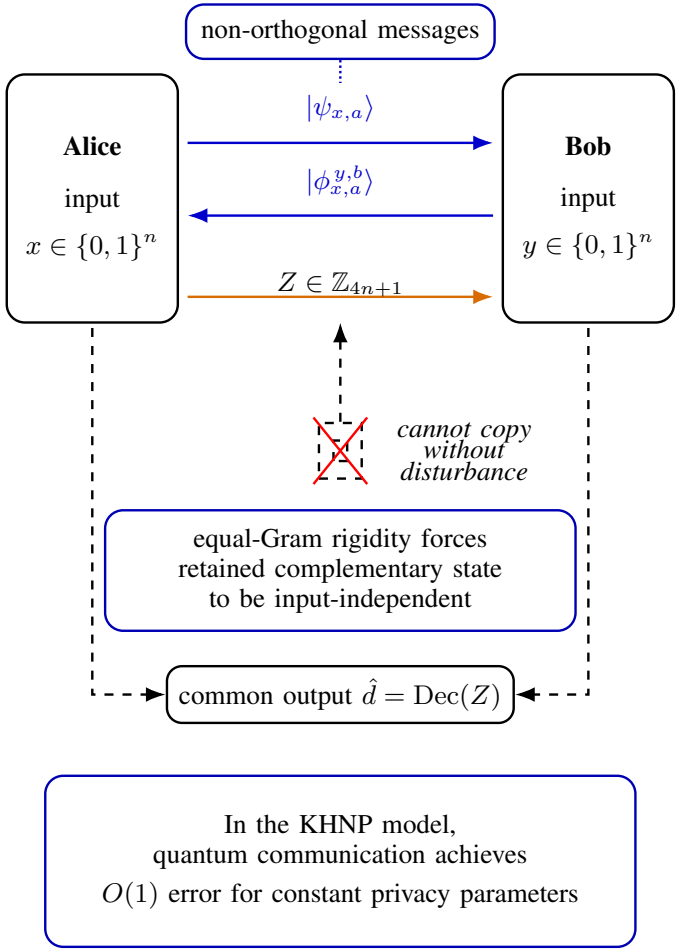
\begin{figure}[t]
\centering
\resizebox{\columnwidth}{!}{%
\begin{tikzpicture}[
    >=Latex,
    font=\small,
    party/.style={
        draw,
        rounded corners=6pt,
        minimum width=2.0cm,
        minimum height=2.9cm,
        align=center,
        thick
    },
    smallbox/.style={
        draw,
        rounded corners=5pt,
        align=center,
        inner sep=4pt,
        thick
    },
    bluebox/.style={
        draw=blue!70!black,
        rounded corners=6pt,
        align=center,
        inner sep=5pt,
        thick
    },
    bigcallout/.style={
        draw=blue!70!black,
        rounded corners=6pt,
        align=center,
        inner sep=6pt,
        thick
    }
]

\node[font=\bfseries\large] at (0,4.8)
{Quantum two-party DP (KHNP model)};

\node[party] (alice) at (-2.9,1.8) {\textbf{Alice}\\[0.25cm]
input\\[0.2cm]
$ x \in \{0,1\}^n $};

\node[party] (bob) at (2.9,1.8) {\textbf{Bob}\\[0.25cm]
input\\[0.2cm]
$ y \in \{0,1\}^n $};

\node[bluebox, minimum width=2.8cm] (nonorth) at (0,3.75) {non-orthogonal messages};

\draw[->, thick, blue!80!black] ($(alice.east)+(0.1,0.65)$) -- ($(bob.west)+(-0.1,0.65)$);
\node[blue!80!black] at (0,2.85) {$|\psi_{x,a}\rangle$};

\draw[<-, thick, blue!80!black] ($(alice.east)+(0.1,-0.2)$) -- ($(bob.west)+(-0.1,-0.2)$);
\node[blue!80!black] at (0,2.0) {$|\phi^{\,y,b}_{x,a}\rangle$};

\draw[densely dotted, thick, blue!80!black] (nonorth.south) -- (0,3.15);

\draw[->, thick, orange!85!black] ($(alice.east)+(0.1,-1.15)$) -- ($(bob.west)+(-0.1,-1.15)$);
\node at (0,0.8) {$Z \in \mathbb{Z}_{4n+1}$};

\node[draw, minimum width=0.5cm, minimum height=0.65cm, thick, dashed] (copyicon) at (0,-1.15) {};
\draw[thick, dashed] ($(copyicon.center)+(-0.08,-0.12)$) rectangle ($(copyicon.center)+(0.08,0.12)$);
\draw[red, thick] ($(copyicon.north west)+(-0.05,0.05)$) -- ($(copyicon.south east)+(0.05,-0.05)$);
\draw[red, thick] ($(copyicon.south west)+(-0.05,-0.05)$) -- ($(copyicon.north east)+(0.05,0.05)$);

\draw[dashed, thick, ->] (copyicon.north) -- (0,0.35);

\node[align=center] at (1.45,-1.15) {\emph{cannot copy}\\[-1mm]\emph{without}\\[-1mm]\emph{disturbance}};

\node[bigcallout, minimum width=5.5cm] at (0,-2.55) {equal-Gram rigidity forces\\ retained
complementary state\\ to be input-independent};

\node[smallbox, minimum width=3.8cm] (out) at (0,-4.0) {common output $\hat d=\mathrm{Dec}(Z)$};

\draw[dashed, thick, ->] ($(alice.south)+(0,-0.05)$) |- (out.west);
\draw[dashed, thick, ->] ($(bob.south)+(0,-0.05)$) |- (out.east);

\node[bigcallout, minimum width=6.9cm, minimum height=2.0cm] at (0,-5.95) {%
In the KHNP model,\\
quantum communication achieves\\[1mm]
$O(1)$ error for constant privacy parameters};

\end{tikzpicture}%
}
\caption{Quantum two-party differential privacy in the KHNP model. Exact preservation of non-orthogonal messages prevents transcript copying and yields the quantum advantage.}
\label{fig:quantum-2party-dp}
\end{figure}

The source of the separation is that classical messages may be copied without
disturbance, so a classical honest observer may retain the complete
transcript~\cite{Vidick_Wehner_2023}. A quantum observer constrained to return the prescribed pure
message cannot do the same. Our construction exploits an
equal-Gram rigidity principle~\cite{CJW04,Cock_2000,Wootters1982}: exact preservation of a family of non-orthogonal
pure messages forces the observer's complementary state to be independent of
the protected input. A small, input-independent coherent guard makes this
rigidity exact, moving the approximation entirely into a rare fallback event
that affects accuracy rather than privacy.

Finally, we delineate the boundary of this advantage. Under a fixed
prescribed-channel view model, reversible classical computation already
realizes the private functionality exactly, so that view model cannot support
a meaningful quantum--classical separation~\cite{KSV02}. At the other extreme, our
protocol is not secure against arbitrary retention-capable 
quantum strategies:
a receiver that measures the first message and aborts can distinguish adjacent
inputs. We formalize the corresponding retention obstruction and provide a
classical randomized-response baseline for this stronger model, which matches the accuracy guarantees of
McGregor et al.~\cite{McGregorMPRTV10,mcgregor2011limits}.
Thus, quantum
communication resolves the Hamming distance accuracy gap under Klauck's exact
message-preserving honesty condition, while the
question of whether it improves accuracy under
fully malicious security remains open.

We consider a range of view models: PC, KHNP, and RR
which stand for \emph{prescribed-channel},
\emph{Klauck's honest nonpreemptive view},
and \emph{retention-robust} views, respectively.
In the PC view model, the parties execute the
fixed local channels specified by the protocol, including their prescribed
discard operations.
Conceptually, KHNP captures an honest-but-curious setting: a party may retain arbitrary private information, but must send the prescribed messages and complete the protocol. In contrast, our retention-robust (RR) model is malicious, allowing arbitrary deviations, including measurement, message modification, retention, and abort.

\subsection{Contributions}

In Table~\ref{tab:status}, we summarize our main results and
conclusions. We work in the standard idealized communication model allowing exact real-valued
private randomness and arbitrary exact local unitaries; local computational complexity
and finite-precision synthesis are not accounted for.
Our main results are as follows.

\begin{enumerate}
\item We show that a common-output functionality, via distributed noise
generation, satisfies \(\eps\)-DP and has error at most \(2/\sinh\eps\).
The proof of privacy and accuracy follows from
the construction via the use of geometric arguments and
interval decoding.
\item Fixed 
prescribed channels (Definition~\ref{def:pc}) can realize this functionality exactly
      with a classical reversible protocol.  In Klauck's stronger model,
      a round trip protocol gives pure \(\eps\)-QDP, error at most
      \(2/\sinh\eps+\gamma\), and \(O(n)\) communication; exact
      noise calibration gives strictly smaller noise
      for a protocol that satisfies \((\eps,\delta)\)-QDP.
\item We define and analyze the
Klauck-honest nonpreemptive view (KHNP) (Definition~\ref{def:khnp}).
The classical version of
KHNP  privacy includes the transcript-copying observer, so
      the classical pure-DP \(\Omega(\sqrt n)\) lower bound transfers and
      gives a same-model quantum accuracy separation.
\item We formalize retention-robust QDP (Definition~\ref{def:rr}), accounting for both adversarial strategies and side information.

\item Randomized response with geometric output noise gives pure
      \(\eps\)-QDP in the retention-robust class (and hence
      \((\eps,\delta)\)-QDP), with an explicit error bound stated in
      Theorem~\ref{thm:rr}.
\end{enumerate}

Figures~\ref{fig:classical-2party-dp} and~\ref{fig:quantum-2party-dp} summarize the differences and similarities between the
protocols that satisfy two-party classical DP versus two-party
quantum DP.

\subsection{Main Results and Implications}

\begin{table*}[!t]
\caption{Summary of the established privacy and error bounds.
We analyze these bounds under different view models:
PC is the Prescribed Channel View (Definition~\ref{def:pc});
KHNP is the Klauck-Honest Nonpreemptive View (Definition~\ref{def:khnp});
RR is the Retention-Robust Adversarial View (Definition~\ref{def:rr}). 
CPTP stands for Completely Positive Trace-Preserving.}
\label{tab:status}
\centering
\footnotesize
\begin{tabular}{@{}p{2.4cm}p{2.4cm}p{11.5cm}@{}}
\toprule
View model & Privacy & Status and utility gap (error)\\
\midrule
PC
& pure \(\eps\)
& Proposition~\ref{prop:pc}: exact classical common output,
\(2/\sinh\eps\).\\
\addlinespace
KHNP
& pure \(\eps\)
& Theorem~\ref{thm:pure-khnp}: quantum common output,
\(<2/\sinh\eps+\gamma\), \(O(n)\) communication.\\
\addlinespace
KHNP
& \((\eps,\delta)\)
& Theorem~\ref{thm:approx-khnp}:
\(\le2/\sinh(\alpha_n^\star)+\gamma\) for some \(\alpha_n^\star > \varepsilon\), \(O(n)\) communication.\\
\addlinespace
Classical KHNP
& pure and approximate
& Corollary~\ref{cor:advantage}: \(\Omega(\sqrt n)\) pure error and
\(\Omega(\sqrt n/\log n)\) error for \(\delta=o(1/n)\), versus \(O(1)\)
quantum error.\\
\addlinespace
RR, all CPTP
& pure \(\eps\)
& Theorem~\ref{thm:rr}:
\(\sqrt n/(2\sinh(\eps/2))+1/(2\sinh^2(\eps/2))\).\\
\addlinespace
RR, all CPTP
& \((\eps,\delta)\)
& The pure protocol of Theorem~\ref{thm:rr} gives the same bound for
every \(\delta\ge0\); improvement is open.\\
\bottomrule
\end{tabular}
\end{table*}

Our main result is an information-theoretic quantum advantage for the
two-party computation of Hamming distance under Klauck's honest,
nonpreemptive, message-preserving privacy model~\cite{klauck2004quantum}. Alice holds
\(x\in\{0,1\}^{n}\), Bob holds \(y\in\{0,1\}^{n}\), and both parties must
output the same estimate of
$d=\HD(x,y).$
The protocol uses quantum communication but requires no computational
assumption, trusted setup, prior entanglement, bounded-storage assumption, or
secure erasure~\cite{BCW98}.

At the \textit{ideal} functional level, Alice and Bob independently sample cyclic
geometric noise shares \(A,B\in\mathbb{Z}_{4n+1}\) and release
\[
Z=d+A+B \pmod{4n+1}.
\]
Both parties apply the same interval decoder to \(Z\). Neither party knows the
other party's noise share, so neither can remove all of the noise from
its own view. The resulting common-output functionality is pure
\(\eps\)-differentially private and has expected error at most
\[
\frac{2}{\sinh\eps}.
\]
The factor of two reflects the common-output requirement for the mistrustful parties: privacy for both
inputs requires both independently generated noise shares to appear in the
same released value.

The quantum protocol implements this functionality through a \textit{guarded coherent
round trip}~\cite{klauck2004quantum}. 
By \textit{guarded}, we mean that every input encoding contains a small common fallback component;
by \textit{coherent}, we mean that Alice sends a superposition rather than a classical input message;
and by \textit{round trip}, we mean that
Bob coherently computes the answer into the message and returns the entire quantum state to Alice.
So, Alice prepares a superposition containing a dominant
data-dependent branch and a small, input-independent \textit{guard branch}. Bob
coherently computes the noisy Hamming distance release, appends it to the
message, and returns the entire prescribed pure state. Alice then measures the
branch and output registers and communicates the resulting classical release
to Bob. Except on the input-independent fallback branch, both parties obtain
the output of the ideal distributed-noise functionality.

\begin{theorem}[Pure QDP for Klauck-honest players. Informal Version of Theorem~\ref{thm:pure-khnp}]
\label{thm:intro-pure-khnp}
For every \(n\ge1\), \(\eps>0\), and \(\gamma>0\), there
exists a protocol that provides
information-theoretic pure \(\eps\)-QDP for
the complete terminal view of every Klauck-honest nonpreemptive observing
strategy.  In an honest execution both parties produce a common output,
\[
 \mathbb E|\widehat d-\HD(x,y)|
 <\frac{2}{\sinh\eps}+\gamma,
\]
and total communication is \(2n+2+O(\log n)=O(n)\) qubits and bits in the
standard unitary communication model.
\end{theorem}

Theorem~\ref{thm:pure-khnp} establishes pure quantum differential privacy for this protocol.
For every \(n\geq 1\), \(\varepsilon>0\), and \(\gamma>0\), the protocol is
pure \(\varepsilon\)-QDP for the complete terminal view of every
Klauck-honest nonpreemptive observing party. It communicates
\[
2n+2+O(\log n)=O(n)
\]
qubits and classical bits and satisfies
\[
\mathbb{E}\!\left[
  \left|\widehat d-\HD(x,y)\right|
\right]
<
\frac{2}{\sinh\varepsilon}+\gamma.
\]
For constant \(\varepsilon\) and arbitrarily small constant \(\gamma\), the
expected error is therefore \(O(1)\), independent of the input length.

This gives a genuine same-model separation from classical
information-theoretic protocols. A classical Klauck-honest party may copy
every classical message while still sending exactly the prescribed outgoing
messages and completing the protocol. Its terminal view can therefore contain
the full classical transcript and private randomness. The lower bounds of
McGregor et al.\ consequently apply under the same honesty requirement:
classical pure-DP protocols for Hamming distance require
\(\Omega(\sqrt n)\) 
error for suitable constant privacy parameters. Hence,
under Klauck's message-preserving model, quantum communication improves the
optimal information-theoretic accuracy from polynomial error to constant
error:
\[
\underbrace{O(1)}_{\text{quantum information-theoretic}}
\qquad\text{versus}\qquad
\underbrace{\Omega(\sqrt n)}_{\text{classical information-theoretic}}.
\]
In particular, the quantum protocol attains information-theoretically the
accuracy of the central model
that classical two-party protocols can obtain through
computationally secure function evaluation~\cite{Haitner22, Haitner25}.

\begin{theorem}[Noise-calibrated approximate KHNP-QDP. Informal Version of Theorem~\ref{thm:approx-khnp}]
\label{thm:intro-approx-khnp}
Let \(n\ge1\), \(\eps>0\), \(0<\delta<1\), and \(\gamma>0\).
There exists a protocol
such that every Klauck-honest nonpreemptive terminal view satisfies
\((\eps,\delta)\)-QDP, both honest parties produce a common output, and
\[
\mathbb E|\widehat d-d|
\le\frac{2}{\sinh\alpha_n^\star}+\gamma
<\frac{2}{\sinh\eps}+\gamma,
\]
for some \(\alpha_n^\star >\varepsilon\).
Communication is \(O(n)\) qubits.
\end{theorem}

Theorem~\ref{thm:approx-khnp} shows that approximate privacy provides a further, strict accuracy
improvement. A direct application of the pure-DP protocol would yield
\((\varepsilon,\delta)\)-QDP with the same noise parameter
\(\alpha=\varepsilon\), leaving the \(\delta\)-budget unused. Instead, we
compute the exact quantum hockey-stick divergence of the finite cyclic
geometric distribution. For every \(0<\delta<1\), this identifies a unique
parameter
\[
\alpha_n^\star(\varepsilon,\delta)>\varepsilon
\]
for which the ideal functionality is exactly
\((\varepsilon,\delta)\)-QDP. Running the same guarded quantum protocol with
noise shares sampled using \(\alpha_n^\star(\varepsilon,\delta)\) yields
\begin{align}
&\mathbb{E}\!\left[
  \left|\widehat d-\HD(x,y)\right|
\right]
\leq\\
&\leq
\frac{2}{\sinh\alpha_n^\star(\varepsilon,\delta)}+\gamma\\
&<
\frac{2}{\sinh\varepsilon}+\gamma.
\end{align}
Thus approximate privacy does not merely inherit the pure-privacy guarantee:
the allowed failure probability is converted directly into a more
concentrated noise distribution and strictly better utility. For strong
approximate privacy, \(\delta=o(1/n)\), the quantum error remains \(O(1)\),
whereas the corresponding classical information-theoretic lower bound is
\(\Omega(\sqrt n/\log n)\).

These results give a model-dependent answer to the accuracy-gap question of
McGregor et al~\cite{McGregorMPRTV10,mcgregor2011limits}. Under exact message-preserving honesty, quantum communication
removes the polynomial accuracy loss suffered by classical
information-theoretic protocols. The advantage does not, however,
automatically extend to arbitrary malicious or retention-capable quantum
strategies. A receiver allowed to measure the first message and abort can
distinguish adjacent inputs. Accordingly, our result should be interpreted as
showing that noncopyable communication is an information-theoretic privacy
resource under exact message-preservation~\cite{Quantum_copying_96}, rather than as a
simulation-secure~\cite{Lindell2017} or maliciously secure two-party protocol~\cite{LindellP07}. Whether quantum
communication improves accuracy under fully retention-robust security remains
open.

\subsection{Main Techniques}

The proof of our main results (Theorem~\ref{thm:pure-khnp} and~\ref{thm:approx-khnp}) combines four ingredients: distributed cyclic-geometric noise,
coherent guarded communication, equal-Gram rigidity, and exact
hockey-stick-divergence calibration.

\paragraph{Distributed noise for a common private output}
The starting point is a common-output functionality in which the two parties
independently sample
\[
A,B\sim q_{\alpha},
\qquad
q_{\alpha}(g)
\propto
\exp\!\left(-\alpha d_m(g,0)\right),
\]
over the cycle \(\mathbb{Z}_{m}\), where \(m=4n+1\) and \(d_m\) is circular
distance. The functionality releases
\[
Z=\HD(x,y)+A+B \pmod m.
\]
Because changing one input bit changes the Hamming distance by at most one,
the ratio between adjacent output probabilities is controlled by the
one-step likelihood ratio of \(q_\alpha\). 

Privacy of Alice's input from Bob is provided by Alice's noise share $A$, which is unknown to Bob.
Likewise, privacy of Bob's input from Alice
is provided by Bob's noise share $B$, 
which is unknown to Alice.

The choice \(m=4n+1\) ensures that the valid distance interval
\(\{0,\ldots,n\}\) is separated from its wraparound copies by a sufficiently
long complementary arc. A nearest-interval decoder
$\Dec(\cdot)$ therefore satisfies
\[
\left|\Dec(Z)-d\right|
\leq
d_m(A+B,0)
\leq
d_m(A,0)+d_m(B,0).
\]
The expected magnitude of each cyclic-geometric share is at most
\(1/\sinh\alpha\), which gives the error bound
\(2/\sinh\alpha\).

\paragraph{A coherent guard that moves approximation into utility}
A direct coherent encoding of the classical input would produce orthogonal
states for different inputs. An observing party could then copy or measure the
input label without disturbing the prescribed outgoing message. To prevent
this, Alice adds a small common component to every encoded state. For
\(i=(x,a)\), she sends
\[
\ket{\psi_i}
=
\sqrt{1-\kappa}\ket{0}\ket{i}
+
\sqrt{\kappa}\ket{1}\ket{\perp},
\]
where the second branch is independent of \(i\). Consequently, for every
\(i\neq j\),
\[
\braket{\psi_i|\psi_j}=\kappa>0.
\]
Bob coherently computes
\[
f_{y,b}(x,a)
=
\HD(x,y)+a+b\pmod m
\]
on the data branch and appends a fixed fallback output on the guard branch,
returning
\[
\ket{\phi^{y,b}_i}
=
\sqrt{1-\kappa}
 \ket{0}\ket{i}\ket{f_{y,b}(i)}
+
\sqrt{\kappa}
 \ket{1}\ket{\perp}\ket{z_\perp}.
\]
The input and returned state families have exactly the same Gram matrix.

The guard probability is chosen as
\[
\kappa=\frac{\gamma}{n+\gamma}.
\]
On the dominant branch, the parties obtain the ideal noisy release. On the
guard branch, they output a fixed fallback value, whose error is at most \(n\).
Thus the guard contributes less than \(\gamma\) to the expected error. The
construction is therefore exactly private for every positive \(\kappa\);
small \(\kappa\) affects utility, not privacy. This avoids taking a
trace-distance limit, which would generally be insufficient for pure QDP
because pure privacy is sensitive to arbitrarily small support mismatches.
\footnote{
Pure differential privacy~\cite{dwork2006calibrating, DworkKMMN06} is equivalent to a bound on the max-divergence.
It is known that the max-divergence between any two distributions $P$ and $Q$ becomes infinite if the support of the first probability distribution $P$ is not contained within the support of the second probability distribution $Q$~\cite{SV15}.
}

\paragraph{Equal-Gram rigidity and the no-information principle}
The central quantum argument is an equal-Gram rigidity lemma~\cite{CJW04,Cock_2000,klauck2004quantum,CHEFLES200014}. 
This is the equal-Gram form of the no-information principle underlying
Klauck's honest-player quantum privacy model~\cite{klauck2004quantum}.
Suppose a Completely Positive Trace-Preserving (CPTP)
map sends each pure input state \(\ket{u_i}\) to a prescribed pure output state
\(\ket{v_i}\), and suppose
\[
\braket{u_i|u_j}
=
\braket{v_i|v_j}
\neq 0
\]
for every \(i,j\). Let \(W\) be a Stinespring isometry for the map~\cite{Stinespring1955}. Purity of
the prescribed output implies
\[
W\ket{u_i}
=
e^{\mathrm{i}\theta_i}
\ket{v_i}\ket{e_i}
\]
for some environment state \(\ket{e_i}\). Preservation of inner products gives
\[
\braket{u_i|u_j}
=
e^{\mathrm{i}(\theta_j-\theta_i)}
\braket{v_i|v_j}
\braket{e_i|e_j}.
\]
Since the two Gram matrices agree and their entries are nonzero,
\(\braket{e_i|e_j}\) must be a phase of unit magnitude. Hence all environment
states agree up to phase, and their density operators are identical. The
observer's retained complementary state is therefore independent of \(i\).

Applied to the guarded protocol, this means that any Klauck-honest Bob who
returns the prescribed pure state cannot retain information about Alice's
pair \((x,a)\). His complete terminal state is merely an input-independent
quantum channel applied to the ideal view \((y,B,Z)\). A symmetric
postprocessing argument handles an observing Alice. Privacy then follows from
the ideal functionality by data processing.

This is precisely where the quantum advantage comes into play. Classical messages can
be copied while remaining unchanged, so exact preservation of a classical
outgoing message places no comparable restriction on the observer's retained
state. For non-orthogonal quantum states, by contrast, exact preservation of
the entire pure-state family prevents the observer from creating an
input-dependent copy in a complementary register.

\paragraph{Exact calibration for approximate privacy}
For approximate privacy, we avoid generic composition or continuity bounds
and calculate the hockey-stick divergence exactly. On the finite cycle, the
likelihood ratio between neighboring shifts of \(q_\alpha\) takes only three
forms:
\[
1,\qquad e^\alpha,\qquad e^{-\alpha}.
\]
The region on which the ratio exceeds \(e^\varepsilon\) is exactly one
half-cycle. Summing the excess probability mass on that region gives
\[
\Delta_{n,\varepsilon}(\alpha)
=
\left(1-e^{\varepsilon-\alpha}\right)
\frac{1-e^{-2n\alpha}}
     {1+e^{-\alpha}-2e^{-(2n+1)\alpha}}.
\]
We prove that \(\Delta_{n,\varepsilon}(\alpha)\) is continuous and strictly
increasing from \(0\) at \(\alpha=\varepsilon\) to \(1\) as
\(\alpha\rightarrow\infty\). Therefore, for every \(0<\delta<1\), there is a
unique solution
\[
\Delta_{n,\varepsilon}
 \!\left(\alpha_n^\star(\varepsilon,\delta)\right)
=
\delta.
\]
Since \(\alpha_n^\star(\varepsilon,\delta)>\varepsilon\), the corresponding
cyclic-geometric distribution is more concentrated than the pure-DP
distribution. The equal-Gram argument is independent of \(\alpha\), so this
improved ideal functionality lifts to the quantum protocol without consuming
any additional \(\delta\)-budget. This yields the strict approximate-privacy
utility improvement in Theorem~\ref{thm:approx-khnp}.

\subsection{Comparison to Previous Lower Bounds}

We now further 
discuss how previous lower bounds in the literature apply
to our work.

\paragraph{The classical transcript model}
A classical interactive protocol induces a random transcript
\[
T=(M_1,\ldots,M_r)
\]
consisting of all messages exchanged by the parties~\cite{GMR85,GMW87,Beaver-1991-Sm,Beaver-1992-FS}. An honest-but-curious
classical player can copy every message it receives or sends without changing
the execution. Its terminal view may therefore contain the complete
transcript, its private input, and its private random coins. Consequently,
privacy of the player's complete view implies privacy of the classical
transcript and coins~\cite{GMR89}.

For every fixed classical transcript \(t\), the probability of observing
\(t\) factors into an Alice-dependent and a Bob-dependent term:
\[
\Pr[T=t\mid X=x,Y=y]
=
a_t(x)b_t(y).
\]
Equivalently, a transcript induces a combinatorial rectangle in the input
space. Conditional on a transcript, the parties' private random coins are
independent. These factorization and conditional-independence properties are
fundamental to classical communication complexity arguments~\cite{kushilevitz1997communication,rao2020communication}.

McGregor et al.\ exploit this classical structure in two complementary ways~\cite{McGregorMPRTV10,mcgregor2011limits}.
Their first technique relates the likelihood-ratio constraints imposed by
differential privacy to deterministic extraction from Santha--Vazirani
sources~\cite{SanthaV86}. Their second technique connects accurate two-party differentially
private protocols to low-communication approximations of the underlying
function~\cite{BARYOSSEF2004702}. Both arguments reason about classical transcript distributions,
transcript rectangles, and the information that can be retained by copying
the interaction.

\paragraph{Why these arguments do not directly extend to quantum protocols}
A quantum protocol does not, in general, generate a simultaneously available
classical transcript~\cite{buhrman1998quantum}. At round \(j\), the communicated register may be in a
non-orthogonal state
\[
\rho^{x,y}_{M_j},
\]
possibly entangled with private workspaces and with registers communicated in
other rounds~\cite{brassard2001quantum}. There need not exist a classical random variable \(T\) whose
realizations describe the interaction while preserving all of its operational
properties.

In particular, the classical step of copying every message is unavailable~\cite{buhrman1998quantum,brassard2001quantum}.
An isometry of the form
\[
\ket{\psi_i}\ket{0}
\longmapsto
\ket{\psi_i}\ket{\psi_i}
\]
cannot exist for a family of distinct non-orthogonal states. Likewise, a
measurement that records which input-dependent state was sent may disturb the
message and prevent the observer from returning the exact state prescribed by
the protocol. Such a measure-and-record strategy~\cite{li2001quantum} is therefore not
Klauck-honest.

This failure of copyability removes the classical reduction from a private
interactive protocol to a private transcript distribution~\cite{Quantum_copying_96}. The observer's
view is instead a density operator, and privacy is expressed through operator
inequalities or quantum hockey-stick divergences~\cite{hirche2023quantum}. There is no general
pointwise factorization
\[
p_{x,y}(t)=a_t(x)b_t(y)
\]
over classical transcript values, and hence no immediate quantum analogue of
the rectangle partition used in the classical lower bound~\cite{buhrman1998quantum,brassard2001quantum}. 
Measuring all
messages to manufacture a classical transcript is also insufficient: the
measurement is an additional channel that can alter interference, destroy
coherence, and violate the exact outgoing-message requirement defining the
KHNP model~\cite{klauck2004quantum}.

More concretely, our guarded messages are deliberately non-orthogonal~\cite{klauck2004quantum,Klauck07}. If
Alice's encoded input and noise share is indexed by \(i=(x,a)\), then the
first-message states satisfy
\[
\braket{\psi_i|\psi_j}=\kappa>0
\qquad\text{for all }i\neq j.
\]
Bob is required to return a prescribed family
\(\{\ket{\phi_i^{y,b}}\}_i\) with exactly the same Gram matrix. Equal-Gram
rigidity then implies that any complementary state retained by a
Klauck-honest Bob is independent of \(i\)~\cite{CJW04}. Thus the action that is automatic
classically (i.e., keeping a copy of the input-dependent message while forwarding
the prescribed message) is physically impossible for these quantum states.
This is because
physical transformations between prescribed families of pure states
can be characterized through their Gram matrices~\cite{CJW04,Watrous2018}. The obstruction can also be seen directly from a Stinespring dilation~\cite{Stinespring1955,CJW04}, as described earlier.

\paragraph{What the classical lower bounds still imply}
The preceding discussion does not mean that the McGregor et al.~\cite{McGregorMPRTV10,mcgregor2011limits} bounds are
irrelevant to our model. They continue to apply to every \emph{classical}
KHNP protocol. A classical KHNP observer may copy every message, retain its
private randomness, send exactly the prescribed outgoing messages, and
complete the protocol. Hence, privacy against all classical KHNP observers
implies privacy of the transcript-and-coins views considered by McGregor et
al~\cite{McGregorMPRTV10,mcgregor2011limits}. Moreover, when honest executions produce the same output with probability
one, that common output is a function of the classical transcript on its
support. Their lower-bound argument therefore yields the classical side of
our separation:
\[
\Omega(\sqrt n)
\quad\text{under pure privacy},
\]
and
\[
\Omega(\sqrt n/\log n)
\quad\text{under strong approximate privacy}.
\]
Theorems~\ref{thm:pure-khnp} and~\ref{thm:approx-khnp} show that these bounds cannot be extended unchanged from
classical to quantum communication.

\paragraph{Dependence on the adversarial model}
The quantum advantage is specific to the exact message-preserving KHNP
condition. It is not a lower-bound violation in a model where the receiver may
arbitrarily measure, modify, retain, or abort. For example, a receiver can
measure the projector onto the data-dependent support of the first message
and then abort, distinguishing adjacent inputs. That strategy is included in
our retention-robust all-CPTP model but excluded from KHNP because it cannot
subsequently return the prescribed coherent state.

Accordingly, prior classical lower bounds and our quantum upper bounds address
different physical realizations, not inconsistent privacy guarantees~\cite{McGregorMPRTV10,mcgregor2011limits,Haitner22,Haitner25}.
Classically, exact message preservation does not prevent transcript retention,
so the classical lower bounds still apply even under KHNP. In the quantum regime, exact
preservation of non-orthogonal pure messages can force the retained
complementary state to contain no input information. Whether an analogue of
the classical accuracy lower bound holds for quantum protocols against fully
retention-capable adversaries remains open.

\subsection{Organization of the Paper}

The remainder of the paper is organized as follows.
In Section~\ref{sec:related-work}, we review related work on classical two-party differential privacy,
quantum privacy models, and stronger cryptographic notions of two-party
security. In Section~\ref{sec:prelims}, we introduce the three view models considered in this
work (fixed prescribed-channel privacy, Klauck-honest nonpreemptive privacy,
and retention-robust adversarial privacy) and discuss the quantum differential
privacy preliminaries used throughout the paper.

In Section~\ref{sec:ideal}, we develop the ideal common-output functionality for Hamming
distance. We introduce the distributed cyclic-geometric noise mechanism,
prove its pure privacy and accuracy guarantees, and derive the exact
finite-cycle hockey-stick divergence needed for approximate privacy.
The corresponding protocol realizations are given in Section~\ref{sec:khnp}. We first show that
the prescribed-channel model admits an exact classical reversible
implementation, and then present the guarded coherent quantum protocol.
Furthermore, we prove the equal-Gram no-information principle, establish the
pure QDP guarantee of Theorem~\ref{thm:pure-khnp}, derive the sharper approximate QDP guarantee
of Theorem~\ref{thm:approx-khnp}, and compare the resulting accuracy of the quantum protocol with the classical
lower bounds of McGregor et al.~\cite{mcgregor2011limits}.

In Section~\ref{sec:discussion},
we discuss the interpretation of the different
privacy models, the distinction between pure and approximate privacy, and the possibility of verification of quantum states
in the KHNP model. Finally, we conclude in Section~\ref{sec:conclusion}.

\section{Related Work}
\label{sec:related-work}

\subsection{Differential privacy in two-party computation}

In two-party differential privacy, Alice and Bob hold private inputs and
interactively compute a statistic of their joint data~\cite{DworkKMMN06}. Privacy must hold in
both directions: Alice's view should be insensitive to a change in Bob's
input, and Bob's view should be insensitive to a change in Alice's input~\cite{BNO08}.
This notion of privacy
differs from the trusted curator model, in which a single trusted
algorithm sees both inputs and releases only a noisy output. For a
sensitivity-one statistic such as Hamming distance, the trusted curator model
achieves expected error \(O(1/\varepsilon)\), independent of the input
length. In the two-party setting, by contrast, the interaction required to
compute the statistic may itself leak information.

McGregor et al.~\cite{mcgregor2011limits} derive
and discuss strong
information-theoretic limitations for classical two-party differential
privacy. For Hamming distance on \(n\)-bit inputs, their results imply an
\(\Omega(\sqrt n)\) error lower bound for suitable constant pure-privacy
parameters and an \(\Omega(\sqrt n/\log n)\) lower bound under strong
approximate privacy. These bounds separate the two-party setting from the
trusted curator model and from computationally private protocols~\cite{BNO08}. Indeed,
standard secure multiparty computation can emulate the trusted curator under
cryptographic assumptions, allowing the parties to recover
accuracy, in the central model,
against computationally bounded observers.

Subsequent work investigated the computational assumptions needed to cross
this information-theoretic accuracy barrier. Haitner et
al.~\cite{Haitner22} showed that sufficiently accurate
computationally differentially private protocols for Hamming distance or
integer inner product imply key agreement. Thus, improving substantially over
the classical information-theoretic bounds is not merely an algorithmic
improvement within the classical model: it entails nontrivial cryptographic
structure. Several other
related works have further studied the relationship between
computational differential privacy, key agreement, oblivious transfer, and
other public-key primitives~\cite{Haitner22,Haitner25,BNO08}.

Our result identifies quantum communication as a different resource for
crossing the classical accuracy barrier. The protocol of Theorem~\ref{thm:pure-khnp} uses no
computational assumption, trusted setup, prior entanglement, bounded-storage
assumption, or secure-erasure assumption. Nevertheless, under Klauck's
honest, nonpreemptive, message-preserving model, it achieves \(O(1)\)
expected error for constant \(\varepsilon\). This does not contradict the
lower bounds of McGregor et al.~\cite{McGregorMPRTV10,mcgregor2011limits}, because those bounds are lower bounds for
\emph{classical} interactive protocols. Their proofs use structural
properties of classical transcripts that do not hold for coherent quantum
communication~\cite{BCW98,brassard2001quantum,buhrman1998quantum}.

\subsection{Quantum privacy models}

Klauck explicitly separated privacy at every round from privacy only in
the final state of honest nonpreemptive parties~\cite{klauck2004quantum}.  The honest players are
defined by exact preservation of every prescribed outgoing message, not
by a fixed sequence of local channels; he also observes that classical
messages may be copied and that final-state privacy is therefore no
weaker classically \cite{klauck2004quantum}.  Our guarded construction
uses the exact no-information principle used in Klauck's work~\cite{klauck2004quantum}.
Kerenidis et al.\ further showed that privacy in quantum communication
depends on whether inputs are classical, coherent, or externally purified
\cite{kerenidis2016privacy}. In our work, the inputs are
classical but quantum communication is allowed. The
classical differential privacy definition has been extended to
allow for more general measurements~\cite{hirche2023quantum}. 
We rely on such quantum differential privacy definitions (which
include the classical differential privacy definition as a special case) in our work.

We note that analyses for stronger adversaries exist in the literature~\cite{KL14}. Such adversaries
may retain auxiliary work while remaining locally consistent with an
honest execution \cite{dupuis2010specious}; even a private SWAP cannot be
implemented in the bare model in that framework~\cite{dupuis2010specious}.  Salvail, Schaffner, and
Sotakova study honest-but-curious quantum parties that purify their actions
and postpone measurement \cite{salvail2015leakage}.  General impossibility
and leakage results for two-party quantum computation imply that
message-preserving honesty \textit{cannot}, without additional
extensions, lead to malicious or simulation-based security
\cite{buhrman2012insecurity,osborn2022constant}.  The retention barrier we discuss in our work
is consistent with these observations.

\section{Models and QDP Preliminaries}
\label{sec:prelims}

Write \(x\sim x'\) when two strings differ in one coordinate.  For density
operators, \(A\preceq B\) denotes the Loewner order and
\[
  E_\lambda(\rho\Vert\sigma)=\Tr(\rho-\lambda\sigma)_+
\]
is the quantum hockey-stick divergence~\cite{Sharma_2013}.

Motivated by the definition of classical differential privacy for
two-party protocols (Definition 2.3 in~\cite{mcgregor2011limits}),
we define quantum differential privacy for two-party protocols:

\begin{definition}[QDP for a family of views]
\label{def:qdp}
A family \(\{\rho^B_{x,y}\}\) is pure \(\eps\)-QDP for Alice's input if,
for every \(x\sim x'\) and fixed \(y\),
\[
 \rho^B_{x,y}\preceq e^\eps\rho^B_{x',y},
 \qquad
 \rho^B_{x',y}\preceq e^\eps\rho^B_{x,y}.
\]
It is \((\eps,\delta)\)-QDP if the corresponding two hockey-stick
divergences with \(\lambda=e^\eps\) are at most \(\delta\).  The symmetric
conditions protect Bob's input in Alice's view.
\end{definition}

These inequalities are equivalent to the usual bounds for every quantum
measurement \(0\preceq M\preceq I\)~\cite{ZhouY17,hirche2023quantum}.  They are preserved by quantum channels,
partial trace, and tensoring input-independent side information.

\begin{definition}[Fixed prescribed-channel view]
\label{def:pc}
In the \emph{prescribed-channel} (PC) view model, the parties execute the
fixed local channels specified by the protocol, including their prescribed
discard operations.  A party's view is the collection of registers
designated by those channels at termination.
\end{definition}

\begin{definition}[Klauck-honest nonpreemptive view]
\label{def:khnp}
Fix a quantum protocol \(\Pi\).  A strategy is \emph{Klauck-honest} if,
for every private input and every sequence of pure message states it could
have received, the density operator of its next outgoing message is
exactly the one prescribed by \(\Pi\).  Subject to this constraint it may
perform arbitrary local operations and retain arbitrary private systems.
It is \emph{nonpreemptive} if it completes every round.  Its view is its
complete terminal private state.  Protocol \(\Pi\) is KHNP-QDP if
Definition~\ref{def:qdp} holds for the terminal views of every such
strategy for the observing party.
\end{definition}

This is Klauck's definition \cite[Sec.~2.3]{klauck2004quantum}.  Copying a
prescribed classical message is allowed.  By contrast, his no-information
lemma shows that no unitary can extract a distinguishing record while
leaving each of two non-orthogonal pure messages unchanged.

The pure KHNP result below uses classical input registers and classical
private random tapes.  The parties may keep their own input and sampled
noise, and honesty is imposed conditional on every realization of that
tape.  The input is not supplied in coherent superposition or entangled
with an external reference.  Subject to these restrictions, the observer
may still use arbitrary quantum private work allowed by
Definition~\ref{def:khnp}.

\begin{definition}[Retention-robust adversarial QDP]
\label{def:rr}
For corrupted Bob, let \(\mathfrak A_B^{\rm all}\) contain all
finite-dimensional interactive CPTP strategies, including strategies
that change messages, retain auxiliary systems, or abort.  A strategy may
start with an arbitrary side register \(R_B\) in state \(\sigma^y_{R_B}\)
that may depend on Bob's fixed input \(y\), but is identical for adjacent
choices \(x,x'\) of Alice's protected input and is independent of Alice's
fresh private coins.  The strategy itself is independent of \(x\) except
through received messages.  Let
\(\rho^{B^*,\sigma}_{x,y}\) be its complete state when it terminates.
RR-QDP for Alice means that Definition~\ref{def:qdp} holds for every
\(B^*\in\mathfrak A_B^{\rm all}\), every \(y,\sigma^y\), and every
\(x\sim x'\).  The symmetric quantifiers, with side information allowed
to depend on \(x\) but identical for adjacent \(y,y'\), protect Bob.
Accuracy and common output are required only in the all-honest execution.
\end{definition}

Allowing initial side information that already changes with the protected
input would make privacy impossible independently of the protocol.
Definition~\ref{def:rr} therefore fixes the same adversarial side state in
the two adjacent experiments.

\section{The Ideal Common-Output Functionality}
\label{sec:ideal}

Fix the odd modulus \(m=4n+1\).  For an integer \(u\), let
\([u]_m\in\{0,\ldots,m-1\}\) be its least nonnegative residue.  The
circular metric on \(\Z_m\) is
\begin{equation}
 d_m(a,b):=\min\{[a-b]_m,[b-a]_m\}
 =\min_{k\in\Z}|\widetilde a-\widetilde b+km|,
\label{eq:circular}
\end{equation}
where \(\widetilde a,\widetilde b\) are arbitrary integer lifts.
For \(\alpha>0\), define the full-support cyclic-geometric distribution
\[
 q_\alpha(g)=
 \frac{\exp[-\alpha d_m(g,0)]}
 {\sum_{h\in\Z_m}\exp[-\alpha d_m(h,0)]}.
\]
Circular distance is \(1\)-Lipschitz, so
\[
 e^{-\alpha}\le\frac{q_\alpha(g+1)}{q_\alpha(g)}
 \le e^\alpha.
\]
Alice samples \(A\sim q_\alpha\), Bob independently samples
\(B\sim q_\alpha\), and the ideal functionality releases
\[
 Z=d+A+B\pmod m.
\]
Let \(D=\{0,\ldots,n\}\subset\Z_m\), and define
\begin{equation}
 \operatorname{Dec}(z):=
 \min\mathop{\arg\min}_{r\in D}d_m(z,r),
\label{eq:decoder}
\end{equation}
so every tie is broken toward the smaller integer.  Both parties receive
the same \(Z\) and output the common estimate
\(\widehat d=\operatorname{Dec}(Z)\).  Alice's ideal view is
\((x,A,Z)\), while Bob's is \((y,B,Z)\).

\begin{lemma}[Circular interval decoding]
\label{lem:decode}
For every \(d\in D\) and \(z\in\Z_m\),
\[
 |\operatorname{Dec}(z)-d|\le d_m(z,d).
\]
\end{lemma}

\begin{proof}
If \(z\in D\), then \(\operatorname{Dec}(z)=z\), and
\(|z-d|=d_m(z,d)\) because \(|z-d|\le n<m/2\).

Suppose \(z\notin D\).  The complementary arc from \(n\) to \(0\) that
avoids the interior of \(D\) has length
\(L=m-n=3n+1\).  Let \(t\in(0,L)\) be the distance along this arc from
\(n\) to \(z\).  For every \(r\in D\), the two circular paths from \(z\)
to \(r\) have lengths
\[
 t+n-r\quad\text{and}\quad L-t+r,
\]
and hence
\[
 d_m(z,r)=\min\{t+n-r,L-t+r\}.
\]
Minimizing over \(r\in D\) shows that the decoder selects \(n\) when
\(t<L/2\), selects \(0\) when \(t>L/2\), and selects \(0\) at equality
by \eqref{eq:decoder}.

If \(\operatorname{Dec}(z)=n\), then
\[
 t+n-d\ge n-d,\qquad
 L-t+d>L/2+d\ge n-d.
\]
If \(\operatorname{Dec}(z)=0\), then
\[
 L-t+d\ge d,\qquad
 t+n-d\ge L/2+n-d\ge d,
\]
where the last inequality uses \(d\le n\) and
\(L/2=(3n+1)/2\).  Taking the minimum of the two path lengths proves the
result in both cases.
\end{proof}

We provide the proof of privacy and accuracy for the use of
geometric noise to satisfy privacy. Our work extends the
analysis of Ghosh et al.~\cite{GRS09} to circular
interval decoding.

\begin{theorem}[Ideal privacy and accuracy]
\label{thm:ideal}
The ideal views are pure \(\eps\)-QDP and
\[
  \mathbb E|\widehat d-d|
  \le \frac{2}{\sinh\eps}\le\frac2\eps.
\]
\end{theorem}

\begin{proof}
For Alice,
\[
 \Pr[A=a,Z=z\mid x,y]
 =q_\eps(a)q_\eps(z-d-a).
\]
Changing one bit of \(y\) changes \(d\) by at most one.  The likelihood
ratio is therefore between \(e^{-\eps}\) and \(e^\eps\).  Exchanging
the parties proves Bob's condition.

Set \(E=A+B\in\Z_m\).  Lemma~\ref{lem:decode}, translation invariance of
\eqref{eq:circular}, and the triangle inequality give
\[
 |\widehat d-d|
 \le d_m(d+E,d)
 =d_m(E,0)
 \le d_m(A,0)+d_m(B,0).
\]

The random variable \(d_m(A,0)\) follows the distribution of
the absolute value of a two-sided geometric random variable, except that it cannot exceed \(2n\). Removing values larger than $2n$ can only decrease its expectation.  
That is, conditioning
on this lower-tail event cannot increase its expectation.  Consequently,
\[
 \mathbb E d_m(A,0)
 \le\frac{2e^{-\eps}}{1-e^{-2\eps}}
 =\frac1{\sinh\eps},
\]
and similarly for \(B\).
\end{proof}

\begin{proposition}[Exact finite-cycle approximate loss]
\label{prop:calibration}
For \(\alpha\ge\eps>0\), define
\begin{equation}
\Delta_{n,\eps}(\alpha)
:=(1-e^{\eps-\alpha})
\frac{1-e^{-2n\alpha}}
 {1+e^{-\alpha}-2e^{-(2n+1)\alpha}}
\label{eq:delta-cal}
\end{equation}
and
\begin{equation}
U_n(\alpha):=
\frac{4\sum_{j=1}^{2n}j e^{-\alpha j}}
 {1+2\sum_{j=1}^{2n}e^{-\alpha j}}.
\label{eq:finite-utility}
\end{equation}
The ideal views are exactly
\((\eps,\Delta_{n,\eps}(\alpha))\)-QDP and satisfy
\[
 \mathbb E|\widehat d-d|
 \le U_n(\alpha)\le\frac{2}{\sinh\alpha}.
\]
Moreover, \(\Delta_{n,\eps}\) is continuous and strictly increasing from
\(0\) at \(\alpha=\eps\) to \(1\) as \(\alpha\to\infty\).  Thus every
\(0<\delta<1\) determines a unique
\(\alpha_n^\star(\eps,\delta)>\eps\) satisfying
\[
 \Delta_{n,\eps}(\alpha_n^\star)=\delta.
\]
\end{proposition}

\begin{proof}
Represent \(\Z_m\) by \(\{-L,\ldots,L\}\), where \(L=2n\), put
\(r=e^{-\alpha}\), and write
\[
 q_\alpha(g)=\frac{r^{|g|}}{1+2\sum_{j=1}^{L}r^j}.
\]
For adjacent distances \(d,d+1\), Alice's probability masses are
\[
 P_d(a,z)=q_\alpha(a)q_\alpha(g),\qquad
 P_{d+1}(a,z)=q_\alpha(a)q_\alpha(g-1),
\]
where \(g=z-d-a\) and subtraction is cyclic.  The relevant ratios are
\[
\frac{q_\alpha(g)}{q_\alpha(g-1)}
=
\begin{cases}
1,&g=-L,\\
e^\alpha,&-L+1\le g\le0,\\
e^{-\alpha},&1\le g\le L.
\end{cases}
\]
Consequently,
\begin{align*}
E_{e^\eps}(P_d\Vert P_{d+1})
&=(1-e^{\eps-\alpha})
  \frac{\sum_{j=0}^{L-1}r^j}
       {1+2\sum_{j=1}^{L}r^j}\\
&=\Delta_{n,\eps}(\alpha).
\end{align*}
The reverse divergence has the same value, supported on the opposite
half-cycle.  Averaging over the observer's own share and exchanging the
parties proves the QDP statement.

Lemma~\ref{lem:decode} and the triangle inequality bound error by the sum
of the two share magnitudes.  Evaluating their finite expectations gives
\eqref{eq:finite-utility}; comparison with the untruncated two-sided
geometric law gives \(U_n(\alpha)\le2/\sinh\alpha\).

For monotonicity, let
\[
A_L(r)=\sum_{j=0}^{L-1}r^j,\qquad
W_L(r)=\frac{A_L(r)}{1+2rA_L(r)}.
\]
Then
\[
W_L'(r)=
\frac{A_L'(r)-2A_L(r)^2}{(1+2rA_L(r))^2}<0,
\]
because \(A_L'(r)\le A_L(r)^2\) coefficientwise for \(0<r<1\).
Both \(W_L(e^{-\alpha})\) and \(1-e^{\eps-\alpha}\) therefore increase
with \(\alpha\), strictly for the latter.  Formula
\eqref{eq:delta-cal} gives the endpoint limits, proving existence and
uniqueness.
\end{proof}

The other party's unknown share is essential: a party that sampled all
the noise could subtract it from its own view.

\begin{remark}[Why the factor two is present]
The main problem requires a common output, so both independent shares
enter the same released \(Z\), giving \(2/\sinh\alpha\).  Without a
common-output requirement, the ideal functionality could instead give
Alice \(d+B\) and Bob \(d+A\) modulo \(m\).  Each party would then decode
using only the other party's share and obtain expected error at most
\(1/\sinh\alpha\).  We do not use that separate-output view model.
\end{remark}

\section{Exact PC and Guarded KHNP Realizations}
\label{sec:khnp}

\begin{proposition}[Exact classical realization in the PC view model]
\label{prop:pc}
In the PC view model of Definition~\ref{def:pc}, there is an exact
classical reversible protocol with \(O(n)\) communication whose common
terminal output and views equal the ideal functionality of
Theorem~\ref{thm:ideal}.  It is therefore pure \(\eps\)-QDP and has
expected error at most \(2/\sinh\eps\).
\end{proposition}

\begin{proof}
Condition on the sampled shares \(a,b\).  Reversibly run a deterministic
computational-basis protocol for
\[
 F_m((x,a),(y,b))=\HD(x,y)+a+b\pmod m.
\]
Copy its value \(Z\) into a permanent output register, copy the same
classical \(Z\) into the other party's output register, and reverse the
computation to clear the message history and work registers.  The fixed
channels of the PC view model perform this reversal and retain no
additional transcript.  The terminal views are exactly \((x,A,Z)\) and
\((y,B,Z)\), and both parties apply the same decoder 
\eqref{eq:decoder}.  Theorem~\ref{thm:ideal} completes the proof.
\end{proof}

Thus PC security alone cannot demonstrate a quantum improvement: it
already permits an exact classical pure-DP realization.

We proceed to the KHNP model where we introduce Algorithm~\ref{alg:khnp-pure}.

\subsection{Guarded coherent round trip}

Set \(\ell=\lceil\log_2m\rceil\), choose a fixed fallback
\(z_\bot\in\Z_m\), and let
\[
 \kappa=\frac{\gamma}{n+\gamma}\in(0,1).
\]
For \(i=(x,a)\), Alice prepares and sends
\begin{equation}
 \ket{\psi_i}
 =\sqrt{1-\kappa}\ket{0}\ket{i}
  +\sqrt{\kappa}\ket{1}\ket{\bot}.
\label{eq:guard-in}
\end{equation}
For fixed \(y,b\), define
\[
 f_{y,b}(x,a)=\HD(x,y)+a+b\pmod m.
\]
Bob coherently appends an output register and returns the entire message:
\begin{equation}
 \ket{\phi_i^{\,y,b}}
 =\sqrt{1-\kappa}\ket{0}\ket{i}\ket{f_{y,b}(i)}
  +\sqrt{\kappa}\ket{1}\ket{\bot}\ket{z_\bot}.
\label{eq:guard-out}
\end{equation}
Alice measures the branch and residue registers, sends the measured
residue to Bob, and both apply \(\operatorname{Dec}\).

\subsection{Privacy}

\begin{theorem}[Pure QDP for Klauck-honest players]
\label{thm:pure-khnp}
For every \(n\ge1\), \(\eps>0\), and \(\gamma>0\), the guarded protocol
with \(A,B\sim q_\eps\) is information-theoretic pure \(\eps\)-QDP for
the complete terminal view of every Klauck-honest nonpreemptive observing
strategy.  In an honest execution both parties produce a common output,
\[
 \mathbb E|\widehat d-\HD(x,y)|
 <\frac{2}{\sinh\eps}+\gamma,
\]
and total communication is \(2n+O(\log n)=O(n)\) qubits and bits in the
standard unitary communication model.
\end{theorem}

\begin{proof}
First fix \(y,b\).  For \(i\ne j\),
\[
 \langle\psi_i|\psi_j\rangle=\kappa
 =\langle\phi_i^{\,y,b}|\phi_j^{\,y,b}\rangle,
\]
while both inner products are \(1\) for \(i=j\).  Any observing Bob must
return exactly the pure state \eqref{eq:guard-out}.  Claim~\ref{clm:gram}
therefore makes his retained complementary state independent of
\(i=(x,a)\), conditional on his own \(y,b\).  He subsequently receives
only the classical residue obtained from the ideal \(Z\) by the channel
\[
 Z\longmapsto
 \begin{cases}
 Z,&\text{with probability }1-\kappa,\\
 z_\bot,&\text{with probability }\kappa.
 \end{cases}
\]
All later local work is postprocessing.  Hence, for every such strategy of Bob, his complete terminal state is a channel output independent of \(x\)
applied to the ideal state \((y,B,Z)\).  Theorem~\ref{thm:ideal} and data
processing prove pure \(\eps\)-QDP for Alice's input, after averaging over
\(B\).

Now fix \(x\) and an observing Alice strategy.  Conditional on \(A=a\),
the state returned by honest Bob depends on \(y,B\) only through
\[
 Z=\HD(x,y)+a+B\pmod m.
\]
For fixed \(x,a\), it is the fixed encoding
\[
\begin{aligned}
 z\longmapsto {}&
 \sqrt{1-\kappa}\ket{0,x,a,z}\\
 &+\sqrt{\kappa}\ket{1,\bot,z_\bot}.
\end{aligned}
\]
Alice's retained preparation registers, arbitrary subsequent instrument,
prescribed final message, and terminal work together form a channel
independent of \(y\) applied to the classical ideal state \((x,A,Z)\).
Data processing of Theorem~\ref{thm:ideal} proves pure
\(\eps\)-QDP for Bob's input.  The two observing strategies were
arbitrary, which falls under the KHNP model.

The branch measurement is input-independent, with guard probability
\(\kappa\).  Conditional on the real branch, Theorem~\ref{thm:ideal}
gives error at most \(2/\sinh\eps\); on the fallback branch the error is
at most \(n\).  Therefore
\[
 \mathbb E|\widehat d-d|
 \le (1-\kappa)\frac{2}{\sinh\eps}+\kappa n
 <\frac{2}{\sinh\eps}+\gamma.
\]
The first message has \(1+n+\ell\) qubits, its return has
\(1+n+2\ell\), and the final residue has \(\ell\) classical bits. 
\end{proof}

\begin{claim}[Equal-Gram no-information principle]
\label{clm:gram}
Let \(\{\ket{u_i}\}\) and \(\{\ket{v_i}\}\) be pure-state families with
\[
 \langle u_i|u_j\rangle=\langle v_i|v_j\rangle\ne0
 \quad\text{for all }i,j.
\]
If a CPTP map (e.g., induced by either Alice or Bob) sends each \(\ketbra{u_i}\) to the pure outgoing state
\(\ketbra{v_i}\), then its complementary output state is independent of
\(i\).
\end{claim}

\begin{proof}
Let \(W\) be a Stinespring isometry~\cite{Stinespring1955}.  Purity of the outgoing marginal
implies
\[
 W\ket{u_i}=e^{\mathrm i\theta_i}\ket{v_i}\ket{e_i}
\]
for some unit vectors \(\ket{e_i}\).  Preservation of inner products and
the nonzero equal-Gram hypothesis give
\[
 \langle e_i|e_j\rangle=e^{\mathrm i(\theta_i-\theta_j)}.
\]
Thus all \(\ket{e_i}\) agree up to phase, and hence all complementary
density operators \(\ketbra{e_i}\) are identical.
\end{proof}

\subsection{Variations}

The use of an exact \(\sqrt{\kappa}\) rotation is possible in the
arbitrary unitary model.  A discrete variant chooses a power of two
\(K\ge\max\{2,n/\gamma\}\) and replaces \eqref{eq:guard-in} by
\[
 \frac1{\sqrt K}\left(\ket{0,\bot}
 +\sum_{r=1}^{K-1}\ket{r,i}\right).
\]
The returned state analogously appends \(z_\bot\) on branch \(0\) and
\(f_{y,b}(i)\) on every branch \(r\ge1\).
The distinct-state overlap and fallback probability are \(1/K\), the
utility penalty is at most \(n/K\le\gamma\), and Hadamards prepare the
branch register exactly.\footnote{
When \(K=2^k\), the uniform branch register
\[
\frac{1}{\sqrt K}\sum_{r=0}^{K-1}\ket r
\]
is prepared exactly by applying \(H^{\otimes k}\) to
\(\ket{0^k}\). Controlled basis-state preparation then produces
\[
\frac{1}{\sqrt K}
\left(
\ket{0}\ket{\perp}
+
\sum_{r=1}^{K-1}\ket r\ket{x,A}
\right)
\]
exactly.
}  Its communication cost is
\(O(n+\log(1/\gamma))\).

Pure QDP already implies \((\eps,\delta)\)-QDP with the same error.
Approximate privacy permits strictly sharper calibration.

\begin{algorithm}
\caption{Guarded Coherent Protocol for Pure KHNP-QDP}
\label{alg:khnp-pure}
\begin{algorithmic}[1]
\Require Alice has \(x\in\{0,1\}^n\); Bob has
\(y\in\{0,1\}^n\); privacy parameter \(\varepsilon>0\);
utility parameter \(\gamma>0\).
\Ensure A common output with pure \(\varepsilon\)-QDP in the
Klauck-honest nonpreemptive model.

\State Set
\[
m\gets 4n+1,
\qquad
\ell\gets\left\lceil\log_2m\right\rceil,
\qquad
\kappa\gets\frac{\gamma}{n+\gamma}.
\]
\State Fix a public fallback residue \(z_\perp\in\Z_m\).

\State \textbf{Alice:} Sample
\[
A\sim q_\varepsilon.
\]
\State \textbf{Bob:} Sample
\[
B\sim q_\varepsilon.
\]

\State Alice prepares a one-qubit branch register \(C\) and an
\((n+\ell)\)-qubit data register \(M\) in the pure state
\[
\ket{\psi_{x,A}}_{CM}
=
\sqrt{1-\kappa}\,
\ket{0}_C\ket{x,A}_M
+
\sqrt{\kappa}\,
\ket{1}_C\ket{\perp}_M,
\]
where \(\ket{\perp}\) is a fixed public computational-basis state.

\State \textbf{Alice \(\rightarrow\) Bob:}
Send the quantum message
\[
M_{A\rightarrow B}
=
\ket{\psi_{x,A}}_{CM}.
\]

\State Bob initializes an \(\ell\)-qubit residue register \(R\)
to \(\ket{0}\).

\State Bob applies a controlled reversible operation satisfying
\[
\ket{0}\ket{x,A}\ket{0}
\longmapsto
\ket{0}\ket{x,A}
\ket{\HD(x,y)+A+B\bmod m}
\]
and
\[
\ket{1}\ket{\perp}\ket{0}
\longmapsto
\ket{1}\ket{\perp}\ket{z_\perp}.
\]

\State The resulting prescribed pure return message is
\[
\begin{aligned}
&\ket{\phi^{y,B}_{x,A}}_{CMR}
={}\\
&
\sqrt{1-\kappa}\,
\ket{0}_C\ket{x,A}_M
\ket{\HD(x,y)+A+B\bmod m}_R
\\
&+
\sqrt{\kappa}\,
\ket{1}_C\ket{\perp}_M\ket{z_\perp}_R.
\end{aligned}
\]

\State \textbf{Bob \(\rightarrow\) Alice:}
Return the entire quantum message
\[
M_{B\rightarrow A}
=
\ket{\phi^{y,B}_{x,A}}_{CMR}.
\]

\State Alice measures \(C\) and \(R\) in the computational basis,
obtaining outcomes
\[
(c,z)\in\{0,1\}\times\Z_m.
\]

\State Set the public released residue to
\[
Z\gets z.
\]
\Comment{When \(c=1\), \(z=z_\perp\).}

\State \textbf{Alice \(\rightarrow\) Bob:}
Send the classical message
\[
M_{A\rightarrow B}^{\mathrm{final}}=Z\in\Z_m.
\]

\State \textbf{Alice:} Output
\[
\widehat d_A\gets\Dec(Z).
\]

\State \textbf{Bob:} Output
\[
\widehat d_B\gets\Dec(Z).
\]

\State \Return
\[
\widehat d_A=\widehat d_B.
\]
\end{algorithmic}
\end{algorithm}

\begin{theorem}[Noise-calibrated approximate KHNP-QDP]
\label{thm:approx-khnp}
Let \(n\ge1\), \(\eps>0\), \(0<\delta<1\), and \(\gamma>0\).
Run the guarded protocol with
\[
 A,B\sim q_{\alpha_n^\star(\eps,\delta)}.
\]
Every Klauck-honest nonpreemptive terminal view satisfies
\((\eps,\delta)\)-QDP, both honest parties produce a common output, and
\[
\mathbb E|\widehat d-d|
\le U_n(\alpha_n^\star)+\gamma
\le\frac{2}{\sinh\alpha_n^\star}+\gamma
<\frac{2}{\sinh\eps}+\gamma.
\]
Communication is \(O(n)\) qubits in the arbitrary unitary model.
\end{theorem}

\begin{proof}
Proposition~\ref{prop:calibration} makes the ideal views at
\(\alpha=\alpha_n^\star\) exactly \((\eps,\delta)\)-QDP.  The equal-Gram
argument in Theorem~\ref{thm:pure-khnp} does not depend on the value of
\(\alpha\): every complete KHNP terminal view is a quantum channel,
independent of the protected input, applied to its corresponding ideal
view.  Data processing of the hockey-stick divergence proves privacy.
The real branch has expected error at most \(U_n(\alpha_n^\star)\);
the input-independent guard branch contributes at most
\(\kappa n<\gamma\).  The communication calculation is unchanged.
\end{proof}

A convenient explicit choice is
\begin{equation}
\bar\alpha(\eps,\delta)
:=\log\frac{e^\eps+\delta}{1-\delta}.
\label{eq:explicit-alpha}
\end{equation}
Indeed, the finite-cycle factor in \eqref{eq:delta-cal} is at most
\((1+e^{-\alpha})^{-1}\), since for \(r=e^{-\alpha}\),
\[
\frac{1-r^{2n}}{1+r-2r^{2n+1}}\le\frac1{1+r}.
\]
Therefore
\[
\Delta_{n,\eps}(\alpha)
\le\frac{e^\alpha-e^\eps}{e^\alpha+1}.
\]
Thus \(\alpha_n^\star\ge\bar\alpha\), and Theorem~\ref{thm:approx-khnp}
implies the fully explicit bound
\begin{equation}
\mathbb E|\widehat d-d|
\le
\frac{4(1-\delta)(e^\eps+\delta)}
{(e^\eps+1)(e^\eps-1+2\delta)}
+\gamma.
\label{eq:explicit-error}
\end{equation}
The exact finite-\(n\) calibration is at least as accurate.

\subsection{Advantage}

\begin{corollary}[Same-model quantum advantages]
\label{cor:advantage}
For every fixed sufficiently small constant \(\eps>0\) and fixed
\(\gamma>0\), common-output Hamming distance has \(O(1)\) expected error
and \(O(n)\) quantum communication  in the KHNP model, whereas every
classical information-theoretic pure-\(\eps\)-DP KHNP protocol has
\(\Omega(\sqrt n)\) error.  For strong approximate privacy
\(\delta=o(1/n)\), the quantum error remains \(O(1)\), whereas the
classical KHNP error is \(\Omega(\sqrt n/\log n)\).
\end{corollary}

\begin{proof}
A classical KHNP observer may copy every received and sent classical
message, retain its private coins, and still send exactly the prescribed
messages and complete the protocol \cite[Sec.~2.3]{klauck2004quantum}.
Thus privacy against every classical KHNP strategy implies privacy of the
standard transcript-and-coins views used by McGregor et al.\
\cite[Def.~2.3]{mcgregor2011limits}.  Moreover, because the parties'
private coins are independent conditional on a classical transcript,
probability-one agreement makes the common output a function of that
transcript on its support.  The pure-DP partition-bound consequence of
McGregor et al.\ then gives the \(\Omega(\sqrt n)\) Hamming distance
error lower bound for suitable constant \(\eps\)
\cite[Sec.~4.2]{mcgregor2011limits}.  A hypothetical \(o(\sqrt n)\)
expected error would, by Markov's inequality, give \(o(\sqrt n)\) error
with constant success probability and contradict that bound.  The
quantum upper bound follows from Theorem~\ref{thm:pure-khnp}.  The same
transcript-view reduction and the approximate lower bound of McGregor
et al.\ give \(\Omega(\sqrt n/\log n)\) error for \(\delta=o(1/n)\), while
Theorem~\ref{thm:approx-khnp} is \(O(1)\) error for fixed \(\eps,\gamma\).
\end{proof}

The separation is specific to exact KHNP honesty.  For the averaged first
message and adjacent \(x,x'\), the projector
\[
 P_x=\sum_{a\in\Z_m}\ketbra{0,x,a}
\]
has probability \(1-\kappa\) on input \(x\) and probability \(0\) on
input \(x'\).  A receiver allowed to measure \(P_x\) and abort therefore
violates pure DP.  Such a measurement cannot be followed by the exact
coherent return \eqref{eq:guard-out}; KHNP excludes it, whereas the
general-CPTP RR model includes it.

\section{Discussion}
\label{sec:discussion}

\subsection{Interpreting the Three View Models}
\label{sec:discussion-models}

The prescribed-channel (PC), Klauck-honest nonpreemptive (KHNP), and
retention-robust (RR) models impose genuinely different requirements on
a party's behavior and retained information. 
The distinction between
these models is essential for interpreting both the positive results and their
limitations.

In the PC model (Definition~\ref{def:pc}), the parties execute the fixed local channels specified
by the protocol, including prescribed uncomputation and discard
operations, and only the designated terminal registers are included in
their views. This view model describes the terminal state of a fixed
honest implementation. It does not allow a party to make additional
copies of intermediate classical messages before the prescribed
uncomputation. Indeed, Proposition~\ref{prop:pc} gives an
exact classical reversible realization of the ideal functionality in
this model: the protocol computes the noisy residue, copies it into the
permanent output registers, reverses the remaining computation, and
retains no transcript. Note that
PC privacy alone cannot support a
quantum--classical separation.

The KHNP model (Definition~\ref{def:khnp}) allows substantially more local freedom. A party may
perform arbitrary local quantum operations and retain arbitrary private
systems, provided that every outgoing message density operator is
exactly the one prescribed by the protocol and that the party completes
the interaction. In the classical setting, this condition permits a
party to copy every received message while continuing to send the
prescribed messages. Thus a classical KHNP observer may retain the
standard view (of the transcript and private coins) appearing in classical
two-party differential-privacy lower bounds~\cite{McGregorMPRTV10,mcgregor2011limits}.

Our guarded coherent construction exploits the fact that the same
conclusion need not hold in the quantum setting. Its incoming messages are
non-orthogonal, and the prescribed returned pure-state family has exactly
the same nonzero Gram matrix. If an observing party returns the exact
prescribed state, Stinespring dilation and preservation of inner
products force the complementary state retained by that party to be
independent of the protected input. The resulting quantum--classical
separation is therefore obtained under the same exact
message-preserving KHNP model: classical messages can be copied
while preserving the protocol, whereas the guarded non-orthogonal
quantum messages cannot be copied in this manner.

The RR model (Definition~\ref{def:rr}) is stronger. It permits arbitrary interactive CPTP
strategies, including strategies that retain extra systems, change
outgoing messages, measure intermediate states, or abort. The guarded
protocol is not RR secure. For example, a receiver may measure a
projector that distinguishes the real branch for one candidate input
and then terminate the protocol. Such a strategy is excluded by KHNP
because it cannot subsequently return the exact coherent state, but it
is included in the all-CPTP RR class of strategies.

The three models therefore can be used to answer \textit{different questions}:
\begin{enumerate}
\item\textbf{PC:}
What information remains in one prescribed honest implementation?
\item\textbf{KHNP:}
What can a party retain while preserving every prescribed
message exactly and completing the protocol?
\item\textbf{RR:}
What privacy remains against arbitrary interactive deviations,
retention, and early termination?
\end{enumerate}

The quantum advantage established in our work is
via the KHNP model. The advantage is
stronger than a comparison based only on prescribed terminal
registers, but it does not imply
malicious or simulation-based security.

\subsection{From KHNP Honesty to Verifiable Message Preservation}
\label{sec:discussion-verification}

Although KHNP provides a clean same-model separation,
its practical motivation might be 
less immediate. A user may reasonably
object that a genuinely malicious counterparty has no obligation to
preserve the prescribed message or complete the protocol. From that
perspective, proving privacy only for KHNP strategies does not by itself
guarantee protection against the adversaries of practical concern~\cite{klauck2004quantum}.

A potentially more compelling interpretation is to supplement the
guarded protocol with a mechanism that allows the sender to
\emph{test} whether the returned quantum state is close to the state
that should have been returned. KHNP would then describe the
detection limit of a broader verifiable protocol: an exactly
message-preserving party is accepted with probability one, while a
party that significantly changes the prescribed return risks being
detected.

Suppose Alice expects the pure returned state
\[
    \ket{\phi_i^{y,b}}.
\]
If Alice can construct the corresponding verification measurement, she
may apply the two-outcome projective test
\[
    \left\{
        \Pi_i^{y,b},
        I-\Pi_i^{y,b}
    \right\},
    \qquad
    \Pi_i^{y,b}
    :=
    \ket{\phi_i^{y,b}}\!\bra{\phi_i^{y,b}}.
\]
For an honestly returned state, the acceptance probability is one. If
Bob instead returns a state \(\rho_i\), the acceptance probability is
\[
    p_{\mathrm{acc}}(i)
    =
    \operatorname{Tr}\!\left(
        \Pi_i^{y,b}\rho_i
    \right)
    =
    \bra{\phi_i^{y,b}}\rho_i\ket{\phi_i^{y,b}},
\]
and the rejection probability is
\[
    p_{\mathrm{rej}}(i)
    =
    1-
    \bra{\phi_i^{y,b}}\rho_i\ket{\phi_i^{y,b}}.
\]
Thus, deviations that significantly reduce the fidelity of the returned state, relative to the prescribed state, are detected with high probability.

This suggests an operational interpretation of the equal-Gram result.
If Bob passes the ideal test with probability one for every possible
input, then he returns the exact prescribed pure-state family. The
equal-Gram no-information principle then implies that his retained
complementary state is independent of Alice's protected input. At the
opposite extreme, a strategy that extracts substantial distinguishing
information is expected to disturb the returned state and thereby
create a nonzero probability of rejection.
Extracting information about a non-orthogonal quantum-state ensemble
creates disturbance, and quantitative information and disturbance
tradeoffs are well studied~\cite{FuchsPeres1996,BuscemiSacchi2007}.
Quantitative mathematical statements
relating retained information, fidelity of the return, and detection
probability could turn this intuition into an approximate,
verifiable analogue of KHNP privacy. We believe this is
important future work.
Verification of a known pure target state can be formulated as a
two-outcome test whose accepting operator contains the target-state
projector~\cite{PallisterLindenMontanaro2018}.

Such verification could provide practical value even when it does not
prevent every privacy violation. First, it could inform a user that a
counterparty may have deviated from the privacy-preserving interaction.
Second, repeated tests, trap rounds, or random audits could discourage
malicious behavior by creating a risk of public detection. This may be
relevant when service providers, laboratories, or network operators
care about reputation, contractual compliance, or continued access to
clients.
Hidden trap states and randomly selected test computations are standard
tools in verifiable delegated quantum computation
\cite{FitzsimonsKashefi2017}.

The quantum setting is important here. Classically, a receiver can copy
a message perfectly and then return or forward an unchanged copy. A
test of the returned classical message cannot reveal that the receiver
retained an additional copy. In contrast, the guarded quantum states
are non-orthogonal. Extracting an input-dependent record while preserving
the exact return is forbidden by the equal-Gram argument. Verification
of the returned quantum state can therefore certify a property that has
no direct classical analogue: consistency between low disturbance and
the absence of an input-dependent retained complement~\cite{PallisterLindenMontanaro2018}.

\section{Conclusion}
\label{sec:conclusion}

We have shown that classical information-theoretic accuracy lower bounds for
two-party differential privacy do not automatically 
extend to models where quantum
communication is allowed. Under Klauck's honest, nonpreemptive, message-preserving model,
two-party Hamming distance admits an information-theoretic quantum protocol
with \(O(n)\) communication and \(O(1)\) expected error for constant privacy
parameters. In the same model, classical protocols require
\(\Omega(\sqrt n)\) error under pure privacy and
\(\Omega(\sqrt n/\log n)\) error under strong approximate privacy.

The separation is obtained through a distributed cyclic-geometric noise
mechanism and a guarded coherent round trip. Exact preservation of the Gram
matrix of the transmitted pure states forces any Klauck-honest observer's
retained complementary state to be independent of the protected input.
For approximate differential privacy, an exact finite-cycle hockey-stick divergence calculation uses
the \(\delta\)-budget to choose a more concentrated noise distribution and
strictly improve accuracy.

The result is inherently model dependent. Fixed prescribed channels already
admit an exact classical reversible realization, whereas the guarded quantum
protocol is not secure against arbitrary retention-capable CPTP strategies,
such as measure-and-abort attacks. We provide a classical
randomized-response baseline for this stronger model in the appendix. Whether quantum
communication can improve retention-robust accuracy remains open.

\section{Acknowledgements}

We are grateful for helpful related discussions
with Noam Mazor, Kaiyan Shi, Manasi Shingane, and Salil Vadhan. Daniel Alabi was supported by the Simons
Foundation (965342, D.A.) as part of the Junior Fellowship from the Simons Society of Fellows. Emil T. Khabiboulline~acknowledges support from the NRC Research Associateship Program at the National Institute of Standards and Technology (NIST), administered by the Fellowships Office of the National Academies of Sciences, Engineering, and Medicine, and was supported in part by ONR MURI (award No. N000142612102).

\bibliographystyle{IEEEtran}
\bibliography{main}

\clearpage

\appendix

The appendix provides expanded constructions and additional details supporting the main results.

We first give a complete description of the ideal common-output functionality for Hamming distance. In this functionality, Alice and Bob independently sample cyclic-geometric noise shares and receive the same noisy release. We explain why two independent shares are needed to protect both parties while producing a common output, and we give the corresponding privacy and accuracy guarantees. We then present an explicit reversible classical realization of this functionality in the prescribed-channel model. This construction clarifies why the prescribed-channel view is too weak to yield a quantum--classical separation: once intermediate computation is uncomputed, the designated terminal views coincide exactly with those of the ideal functionality.
We next give the full noise-calibrated guarded protocol for approximate QDP. The construction uses the same coherent communication procedure as the pure QDP protocol, but chooses the cyclic-geometric parameter using the exact finite-cycle hockey-stick divergence. This allows the available \(\delta\)-budget to be converted directly into a more concentrated noise distribution, yielding strictly smaller error without changing the communication complexity or the equal-Gram privacy argument. We also provide numerical examples illustrating the utility improvement over the pure QDP calibration. Finally, we give a common-output protocol for the stronger retention-robust model. The protocol combines randomized response for Alice's input with geometric output noise for Bob's input and remains private against arbitrary retention-capable CPTP strategies. Its expected error is \(O(\sqrt{n})\) for constant privacy parameters, providing a classical retention-robust baseline against which future quantum protocols can be compared.

\subsection{Implementing the Ideal Functionality}

In general,
cryptographic protocols are designed to be \emph{realizations} of ideal functionalities, often conceptualized through \emph{ideal views}. In modern cryptography, security is proven by showing that a real-world protocol emulates an ideal world where a trusted third party safely handles the task~\cite{KL14}.

Theorem~\ref{thm:ideal} gives the privacy and accuracy
guarantees for the ideal views for solving the Hamming
distance problem. We now describe, in more detail, how
to implement the common-output distributed-noise ideal
functionality. Throughout, 
we denote this functionality as $\mathcal F$.

\begin{algorithm}
\caption{Ideal Common-Output Distributed-Noise Functionality}
\label{alg:ideal-common-output}
\begin{algorithmic}[1]
\Require Alice has \(x\in\{0,1\}^n\); Bob has
\(y\in\{0,1\}^n\); parameter \(\eps>0\).
\Ensure Alice and Bob receive the same estimate of
\(d=\HD(x,y)\).

\State Set \(m\gets 4n+1\).
\State Define the cyclic distance
\[
d_m(u,v)
:=
\min\{[u-v]_m,[v-u]_m\}.
\]
\State Define the cyclic-geometric distribution
\[
q_\eps(g)
:=
\frac{\exp(-\eps d_m(g,0))}
     {\sum_{h\in\Z_m}\exp(-\eps d_m(h,0))}
\qquad (g\in\Z_m).
\]

\Statex
\State \textbf{Alice:} Sample \(A\sim q_\eps\).
\State \textbf{Bob:} Sample \(B\sim q_\eps\).

\State \textbf{Alice \(\rightarrow\mathcal{F}\):}
Send the classical pair
\[
M_{A\rightarrow\mathcal F}=(x,A)
\]
over Alice's ideal private interface.

\State \textbf{Bob \(\rightarrow\mathcal{F}\):}
Send the classical pair
\[
M_{B\rightarrow\mathcal F}=(y,B)
\]
over Bob's ideal private interface.

\State \textbf{Functionality \(\mathcal F\):} Compute
\[
Z
\gets
\HD(x,y)+A+B
\pmod m.
\]

\State \textbf{\(\mathcal F\rightarrow\) Alice:}
Send the classical residue
\[
M_{\mathcal F\rightarrow A}=Z\in\Z_m.
\]

\State \textbf{\(\mathcal F\rightarrow\) Bob:}
Send the same classical residue
\[
M_{\mathcal F\rightarrow B}=Z\in\Z_m.
\]

\State Define
\[
\Dec(z)
:=
\min\arg\min_{r\in\{0,\ldots,n\}}d_m(z,r),
\]
where ties are broken toward the smaller integer.

\State \textbf{Alice:} Output
\[
\widehat d_A\gets\Dec(Z).
\]

\State \textbf{Bob:} Output
\[
\widehat d_B\gets\Dec(Z).
\]

\State \Return
\[
\widehat d_A=\widehat d_B=\Dec(Z).
\]
\end{algorithmic}
\end{algorithm}

\paragraph{Technical Explanation of Algorithm~\ref{alg:ideal-common-output}}
The ideal functionality separates the privacy noise into two independently
generated shares. Alice knows \(A\), Bob knows \(B\), and both parties receive
only the common release
\[
Z=\HD(x,y)+A+B \pmod m,
\qquad m=4n+1.
\]
The use of two noise shares is essential for a common-output protocol. If
Alice generated all of the noise, she could subtract it from \(Z\) and recover
the exact Hamming distance; the analogous problem would arise if Bob generated
all of the noise. Instead, each party knows only one share and must rely on the
other party's unknown share for privacy.

For example, Alice's ideal view is
\[
V_A^{\mathrm{ideal}}=(x,A,Z).
\]
Conditioned on \(A=a\), its probability mass is
\[
\Pr[A=a,Z=z\mid x,y]
=
q_\eps(a)\,
q_\eps\!\left(z-\HD(x,y)-a\right).
\]
Changing one bit of \(y\) changes \(\HD(x,y)\) by at most one. Since the
cyclic-geometric distribution satisfies
\[
e^{-\eps}
\leq
\frac{q_\eps(g+1)}{q_\eps(g)}
\leq
e^\eps,
\]
the likelihood ratio between Alice's views on adjacent values of \(y\) is at
most \(e^\eps\). The symmetric argument protects Alice's input in Bob's
view. Notice that the privacy parameter is \(\eps\), rather than
\(2\eps\): for either observer, only the other party's unknown noise share
must absorb the unit change in Hamming distance.

The modulus \(m=4n+1\) is chosen so that the valid distance interval
\[
D=\{0,\ldots,n\}
\]
occupies a sufficiently short arc of the cycle. The nearest-point decoder
therefore satisfies
\[
\left|\Dec(Z)-\HD(x,y)\right|
\leq
d_m(A+B,0)
\leq
d_m(A,0)+d_m(B,0).
\]
Each cyclic-geometric share has expected circular magnitude at most
\(1/\sinh\eps\), and hence
\[
\mathbb{E}
\left[
\left|\Dec(Z)-\HD(x,y)\right|
\right]
\leq
\frac{2}{\sinh\eps}.
\]
The factor of two is a utility cost of requiring the parties to output the
same private release: both independently generated shares appear in the
common value \(Z\).

\subsection{Exact Classical Reversible Protocol}

The prescribed-channel (PC) realization shows that the ideal common-output functionality can already be implemented exactly using classical reversible computation. Alice samples her noise share \(A\), Bob samples his noise share \(B\), and the parties reversibly compute

\[
Z = \mathrm{HD}(x,y) + A + B \pmod m .
\]

Once \(Z\) has been computed, it is copied into permanent output registers for both parties. The parties then run the reversible computation backward, clearing the temporary message and workspace registers while leaving the output registers unchanged. At termination, Alice's view is therefore exactly \((x,A,Z)\), while Bob's view is exactly \((y,B,Z)\), matching the views of the ideal functionality.

The important point is that reversibility does not make the intermediate messages private. While the computation is being carried out, the computational-basis messages may reveal substantial information about the parties' inputs. The PC model simply does not include copies of these intermediate messages in the final view unless the prescribed protocol explicitly retains them. Because the protocol uncomputes its temporary state, only the intended noisy output remains. Thus, in the PC model, the classical protocol inherits exactly the privacy and accuracy guarantees of the ideal functionality.

This construction also explains why the PC model cannot by itself demonstrate a quantum advantage. Classically, the parties can compute the functionality, copy out the noisy result, and erase the remaining transcript by reversible uncomputation. The distinction with the KHNP model is that a KHNP observer is allowed to retain additional information while still sending the prescribed messages. A classical KHNP observer can therefore copy the interaction transcript before it is uncomputed, whereas the PC model only evaluates the registers left by the prescribed implementation. This difference is why the classical lower bounds apply in the KHNP model but not directly to the weaker PC realization.

\begin{algorithm}
\caption{Exact Classical Reversible Realization in the
Prescribed-Channel Model}
\label{alg:pc-reversible}
\begin{algorithmic}[1]
\Require Alice has \(x\in\{0,1\}^n\); Bob has
\(y\in\{0,1\}^n\); parameter \(\eps>0\).
\Ensure Both parties output the same estimate of \(\HD(x,y)\).

\State Set
\[
m\gets 4n+1,
\qquad
\ell\gets\left\lceil\log_2m\right\rceil.
\]

\State \textbf{Alice:} Sample \(A\sim q_\eps\).
\State \textbf{Bob:} Sample \(B\sim q_\eps\).

\State Alice initializes an \((n+\ell)\)-bit message register
\(M_{XA}\) to \(0\).
\State Alice reversibly copies her classical input and noise share:
\[
\ket{x,A}_{XA}\ket{0}_{M_{XA}}
\longmapsto
\ket{x,A}_{XA}\ket{x,A}_{M_{XA}}.
\]

\State \textbf{Alice \(\rightarrow\) Bob:}
Send the computational-basis message
\[
M_{XA}=\ket{x,A}.
\]

\State Bob initializes an \(\ell\)-bit temporary register
\(T\) to \(\ket{0}\).

\State Bob reversibly computes
\[
\ket{x,A}_{M_{XA}}
\ket{y,B}_{YB}
\ket{0}_{T}
\longmapsto
\ket{x,A}_{M_{XA}}
\ket{y,B}_{YB}
\ket{Z}_{T},
\]
where
\[
Z=\HD(x,y)+A+B\pmod m.
\]

\State Bob initializes a permanent output register \(Z_B\) to \(0\)
and copies the result:
\[
\ket{Z}_{T}\ket{0}_{Z_B}
\longmapsto
\ket{Z}_{T}\ket{Z}_{Z_B}.
\]

\State Bob initializes an \(\ell\)-bit message register \(M_Z\)
to \(0\) and copies
\[
\ket{Z}_{T}\ket{0}_{M_Z}
\longmapsto
\ket{Z}_{T}\ket{Z}_{M_Z}.
\]

\State \textbf{Bob \(\rightarrow\) Alice:}
Send the computational-basis message
\[
M_Z=\ket{Z}.
\]

\State Alice initializes a permanent output register \(Z_A\) to \(0\)
and copies
\[
\ket{Z}_{M_Z}\ket{0}_{Z_A}
\longmapsto
\ket{Z}_{M_Z}\ket{Z}_{Z_A}.
\]

\State Alice reversibly clears the received message register:
\[
\ket{Z}_{M_Z}\ket{Z}_{Z_A}
\longmapsto
\ket{0}_{M_Z}\ket{Z}_{Z_A}.
\]

\State Bob runs his computation of \(Z\) backward, clearing \(T\)
and every temporary work register while leaving \(Z_B\) unchanged.

\State \textbf{Bob \(\rightarrow\) Alice:}
Return the computational-basis register
\[
M_{XA}=\ket{x,A}.
\]

\State Alice reversibly clears this returned copy using her retained
register \(XA\):
\[
\ket{x,A}_{XA}\ket{x,A}_{M_{XA}}
\longmapsto
\ket{x,A}_{XA}\ket{0}_{M_{XA}}.
\]

\State \textbf{Alice:} Output
\[
\widehat d_A\gets\Dec(Z_A).
\]

\State \textbf{Bob:} Output
\[
\widehat d_B\gets\Dec(Z_B).
\]

\State \Return
\[
\widehat d_A=\widehat d_B.
\]
\end{algorithmic}
\end{algorithm}

\paragraph{Technical Explanation of Algorithm~\ref{alg:pc-reversible}}
Algorithm~\ref{alg:pc-reversible} is an exact realization of the ideal
functionality in the prescribed-channel model. The parties first evaluate
\[
F_m((x,A),(y,B))
=
\HD(x,y)+A+B\pmod m
\]
by reversible classical computation. Once the residue \(Z\) has been
computed, it is copied into permanent output registers \(Z_A\) and \(Z_B\).
The parties then reverse the computation, returning every temporary message
and workspace register to its fixed initial state.

The role of reversibility is not to make the intermediate messages private.
Indeed, intermediate computational-basis messages may reveal substantial
information while they are present. Rather, reversibility ensures that,
under the prescribed-channel view model, the designated terminal views
contain no message history:
\[
V_A^{\mathrm{PC}}=(x,A,Z_A),
\qquad
V_B^{\mathrm{PC}}=(y,B,Z_B).
\]
These are exactly the views supplied by the ideal functionality. The pure
\(\eps\)-DP and \(2/\sinh\eps\) error guarantees therefore follow directly
from the ideal analysis.

This construction also explains why the prescribed-channel model cannot
support a meaningful quantum--classical separation. The model requires the
parties to apply the specified uncomputation and discard operations and scores
only the resulting designated terminal registers. An observer is not allowed
to retain an additional copy of the classical transcript unless that copy is
part of the prescribed channel. Consequently, the standard classical
two-party lower bounds, which score transcript-and-coins views, do not apply
directly to this terminal-view model.

This point is specific to the prescribed-channel model. Under Klauck's
honesty model, a classical player may copy every message, retain the
copy, and still send the prescribed outgoing messages. Thus, the same
classical lower bounds do apply to classical KHNP protocols. Proposition~\ref{prop:pc}
shows that it is necessary to distinguish these two notions of an honest
view before claiming a quantum advantage.

\subsection{Noise-Calibrated Guarded Protocol for Approximate QDP}

The approximate QDP protocol uses the same guarded coherent communication structure as the pure QDP construction in the main body; the key difference is how the cyclic-geometric noise is calibrated. If the pure protocol were run with noise parameter \(\alpha=\varepsilon\), it would automatically satisfy \((\varepsilon,\delta)\)-QDP for any \(\delta>0\), but the additional \(\delta\)-budget would not be used to improve accuracy. So to use up this budget, the protocol computes the exact finite-cycle hockey-stick divergence between neighboring shifts of the cyclic-geometric distribution and chooses a larger parameter \(\alpha_n^\star(\varepsilon,\delta)>\varepsilon\) that exactly meets the desired \((\varepsilon,\delta)\)-privacy level. Because larger \(\alpha\) makes the cyclic-geometric distribution more concentrated around zero, Alice and Bob can add less noise while preserving the target privacy guarantee.

The guarded quantum realization itself is unchanged. Alice and Bob sample independent noise shares from the more concentrated distribution \(q_{\alpha_n^\star}\), and Alice sends the same guarded non-orthogonal encoding used in the pure QDP protocol. Bob coherently appends the noisy Hamming distance residue and returns the prescribed pure state. The equal-Gram argument continues to apply because it does not depend on the value of the noise parameter. Consequently, each party's complete KHNP terminal view is still obtained by input-independent postprocessing of the corresponding ideal view, so the hockey-stick divergence can only decrease under the quantum realization. In particular, no additional \(\delta\)-loss is introduced by implementing the ideal functionality in the quantum setting.

The resulting protocol therefore converts the approximate-privacy budget directly into improved utility. Its expected error is bounded by
\[
U_n\left(\alpha_n^\star(\varepsilon,\delta)\right)+\gamma
\le
\frac{2}{\sinh\left(\alpha_n^\star(\varepsilon,\delta)\right)}+\gamma
<
\frac{2}{\sinh \varepsilon}+\gamma,
\]
while retaining \(O(n)\) communication and the same exact message-preserving privacy argument as in the pure QDP construction. Thus, approximate privacy provides a genuine accuracy improvement (rather than merely inheriting the pure QDP guarantee).

\begin{algorithm}
\caption{Noise-Calibrated Guarded Protocol for Approximate KHNP-QDP}
\label{alg:khnp-approximate}
\begin{algorithmic}[1]
\Require Alice has \(x\in\{0,1\}^n\); Bob has
\(y\in\{0,1\}^n\); privacy parameters
\(\varepsilon>0\) and \(0<\delta<1\);
utility parameter \(\gamma>0\).
\Ensure A common output with \((\varepsilon,\delta)\)-QDP in the
Klauck-honest nonpreemptive model.

\State Set
\[
m\gets 4n+1,
\qquad
\ell\gets\left\lceil\log_2m\right\rceil,
\qquad
\kappa\gets\frac{\gamma}{n+\gamma}.
\]
\State Fix a public fallback residue \(z_\perp\in\Z_m\).

\State Define, for \(\alpha\geq\varepsilon\),
\[
\Delta_{n,\varepsilon}(\alpha)
:=
\left(1-e^{\varepsilon-\alpha}\right)
\frac{1-e^{-2n\alpha}}
     {1+e^{-\alpha}-2e^{-(2n+1)\alpha}}.
\]

\State Let \(\alpha_n^\star(\varepsilon,\delta)\) be the unique
solution to
\[
\Delta_{n,\varepsilon}
\left(
\alpha_n^\star(\varepsilon,\delta)
\right)
=
\delta.
\]

\State Set
\[
\alpha\gets\alpha_n^\star(\varepsilon,\delta).
\]

\State \textbf{Alice:} Sample \(A\sim q_\alpha.\)
\State \textbf{Bob:} Sample \(B\sim q_\alpha.\)

\State Alice prepares
\[
\ket{\psi_{x,A}}_{CM}
=
\sqrt{1-\kappa}\,
\ket{0}_C\ket{x,A}_M
+
\sqrt{\kappa}\,
\ket{1}_C\ket{\perp}_M.
\]

\State \textbf{Alice \(\rightarrow\) Bob:}
Send the quantum message
\[
M_{A\rightarrow B}
=
\ket{\psi_{x,A}}_{CM}.
\]

\State Bob initializes an \(\ell\)-qubit residue register \(R\)
to \(\ket{0}\).

\State Bob applies a controlled reversible operation satisfying
\[
\ket{0}\ket{x,A}\ket{0}
\longmapsto
\ket{0}\ket{x,A}
\ket{\HD(x,y)+A+B\bmod m}
\]
and
\[
\ket{1}\ket{\perp}\ket{0}
\longmapsto
\ket{1}\ket{\perp}\ket{z_\perp}.
\]

\State The resulting prescribed return state is
\[
\begin{aligned}
&\ket{\phi^{y,B}_{x,A}}_{CMR}
={}&\\
&\sqrt{1-\kappa}\,
\ket{0}_C\ket{x,A}_M
\ket{\HD(x,y)+A+B\bmod m}_R
\\
&+
\sqrt{\kappa}\,
\ket{1}_C\ket{\perp}_M\ket{z_\perp}_R.
\end{aligned}
\]

\State \textbf{Bob \(\rightarrow\) Alice:}
Return the entire quantum message
\[
M_{B\rightarrow A}
=
\ket{\phi^{y,B}_{x,A}}_{CMR}.
\]

\State Alice measures \(C\) and \(R\) in the computational basis,
obtaining \((c,z)\).

\State Set
\[
Z\gets z.
\]
\Comment{On the fallback branch, \(Z=z_\perp\).}

\State \textbf{Alice \(\rightarrow\) Bob:}
Send the classical residue
\[
M_{A\rightarrow B}^{\mathrm{final}}=Z\in\Z_m.
\]

\State \textbf{Alice:} Output
\[
\widehat d_A\gets\Dec(Z).
\]

\State \textbf{Bob:} Output
\[
\widehat d_B\gets\Dec(Z).
\]

\State \Return
\[
\widehat d_A=\widehat d_B.
\]
\end{algorithmic}
\end{algorithm}

\paragraph{Numerical savings from approximate privacy}
The approximate QDP calibration permits a larger cyclic-geometric
parameter
\[
    \overline{\alpha}(\varepsilon,\delta)
    =
    \log\left(
        \frac{e^\varepsilon+\delta}{1-\delta}
    \right)
    >
    \varepsilon,
\]
and therefore a smaller noise bound. 

For convenience, denote the pure and approximate QDP expected error bounds as $E_{\mathrm{pure}}(\varepsilon)$ and $E_{\mathrm{approx}}(n,\varepsilon,\delta)$.
Ignoring the arbitrarily small additive guard penalty $\gamma$, the
approximate QDP calibration strictly improves on the pure-QDP bound:
\[
E_{\mathrm{approx}}(n,\varepsilon,\delta)
<
E_{\mathrm{pure}}(\varepsilon),
\]
since
\[
\alpha_n^\star(\varepsilon,\delta)>\varepsilon.
\]

Ignoring the common arbitrarily
small additive guard penalty \(\gamma\), the pure and approximate bounds
are 
\[
    E_{\mathrm{pure}}
    \le
    \frac{2}{\sinh\varepsilon},
    \qquad
    E_{\mathrm{approx}}
    \le
    \frac{2}{
        \sinh\!\left(
            \overline{\alpha}(\varepsilon,\delta)
        \right)
    }.
\]
The exact finite-\(n\) calibration
\(\alpha_n^\star(\varepsilon,\delta)\) is at least as accurate as this
explicit bound.

For example, let \(\varepsilon=0.5\) and \(\delta=0.01\). Pure privacy
gives
\[
    E_{\mathrm{pure}}
    \le
    \frac{2}{\sinh(0.5)}
    \approx 3.838.
\]
In the approximate regime,
\[
    \overline{\alpha}(0.5,0.01)
    =
    \log\left(
        \frac{e^{0.5}+0.01}{0.99}
    \right)
    \approx 0.5161,
\]
and hence
\[
    E_{\mathrm{approx}}
    \le
    \frac{2}{\sinh(0.5161)}
    \approx 3.708.
\]
Thus the explicit approximate-DP calibration lowers the expected-error
upper bound by approximately
\[
    3.838-3.708\approx 0.130,
\]
corresponding to a reduction of about \(3.4\%\).

The savings can be more pronounced in a stronger-privacy regime. For
\(\varepsilon=0.1\) and \(\delta=0.01\), pure privacy gives
\[
    E_{\mathrm{pure}}
    \le
    \frac{2}{\sinh(0.1)}
    \approx 19.967.
\]
Approximate privacy permits
\[
    \overline{\alpha}(0.1,0.01)
    =
    \log\left(
        \frac{e^{0.1}+0.01}{0.99}
    \right)
    \approx 0.1191,
\]
which gives
\[
    E_{\mathrm{approx}}
    \le
    \frac{2}{\sinh(0.1191)}
    \approx 16.759.
\]
The reduction is therefore approximately
\[
    19.967-16.759\approx 3.208,
\]
or about \(16.1\%\).

For small \(\varepsilon\) and \(\delta\),
\[
    \overline{\alpha}(\varepsilon,\delta)
    \approx
    \varepsilon+2\delta,
\]
so the leading comparison is approximately
\[
    E_{\mathrm{pure}}
    \approx
    \frac{2}{\varepsilon},
    \qquad
    E_{\mathrm{approx}}
    \approx
    \frac{2}{\varepsilon+2\delta}.
\]
Thus approximate privacy can provide a meaningful constant-factor
reduction when \(\delta\) is nonnegligible relative to \(\varepsilon\).
In the strong approximate-privacy regime
\(\delta=o(1/n)\), however, this numerical saving is lower order, and
both quantum bounds remain \(O(1)\) for constant \(\varepsilon\).

\subsection{A Common-Output Retention-Robust Baseline}

Quantum communication can always be used to lift
a classical protocol into the quantum setting.
Slight extensions of
previous works (e.g.,~\cite{McGregorMPRTV10, mcgregor2011limits}) can be used to obtain the guarantees
we outline in this section, since every DP algorithm is a
QDP algorithm. However, previous works do not use the exact
constructions we use in this work.
Thus, for the sake of completeness and ease of comparison,
we present our classical protocol here.

We define a clipping operator that will be used in this
section. For $a \le b$, define the \emph{clipping} operator
\[
\operatorname{clip}_{[a,b]}(z)
:=
\min\{b,\max\{a,z\}\}.
\]
Equivalently,
\[
\operatorname{clip}_{[a,b]}(z)
=
\begin{cases}
a, & z<a,\\
z, & a\le z\le b,\\
b, & z>b.
\end{cases}
\]

The
following asymmetric protocol uses randomized response for Alice's input
and geometric output noise for Bob's input; its released estimate is
common to both parties.

\begin{theorem}[Common-output pure RR-QDP baseline]
\label{thm:rr}
For every \(n\ge1\) and \(\eps>0\), there is a two-message protocol using
\(n+O(\log n)\) classical bits (or computational-basis qubits) whose
honest outputs agree with probability one and which is pure
\(\eps\)-RR-QDP under all quantifiers of Definition~\ref{def:rr}.  Its
expected absolute error is at most
\[
 \frac{\sqrt n}{2\sinh(\eps/2)}
 +\frac{1}{2\sinh^2(\eps/2)}.
\]
\end{theorem}

\begin{proof}
Let \(p=(1+e^\eps)^{-1}\) and
\(s=1-2p=\tanh(\eps/2)\).  Independently for each coordinate, Alice sends
\(X_i'=x_i\) with probability \(1-p\) and \(1-x_i\) with probability
\(p\).  Bob computes
\[
 C=\sum_i\mathbf1\{X_i'\ne y_i\},
 \qquad
 \widetilde d=\frac{C-np}{s}.
\]
He samples an independent two-sided geometric random variable
\[
 \Pr[G=g]=\frac{1-e^{-\eps}}{1+e^{-\eps}}e^{-\eps|g|}
 \quad(g\in\Z),
\]
sets \(\widehat d\) to the clipping of
\(\widetilde d+G/s\) to \([0,n]\), and sends a finite encoding of this
clipped value to Alice.  Both parties output that value.

Fix any corrupted-Bob strategy and initial side state permitted by
Definition~\ref{def:rr}.  Alice's only input-dependent transmission is
the randomized-response vector.  A one-bit change of \(x\) changes its
likelihood by at most \(e^\eps\); the adversary's arbitrary CPTP
interaction, retention, or early termination is postprocessing.  Hence
Bob's complete state satisfies pure \(\eps\)-QDP.

Now fix any corrupted-Alice strategy.  Bob first measures the purported
vector in the computational basis, obtaining some \(X'\).  The induced
classical--quantum state of \(X'\) and Alice's retained register is
independent of \(y\).  Conditional on every \(X'\), changing one bit of
\(y\) changes \(C\) by at most one.  The geometric mechanism therefore
obeys the likelihood-ratio bound \(e^\eps\) pointwise in \(X'\).
Multiplying each pointwise inequality by the same conditional retained
state and summing proves the operator inequality for Alice's joint final
view.  This proves pure RR-QDP in both directions, for every allowed
initial side state.

In an honest run, \(\widetilde d\) is unbiased and
\[
 \operatorname{Var}(\widetilde d)
 =\frac{np(1-p)}{s^2}
 =\frac{n}{4\sinh^2(\eps/2)}.
\]
Moreover, \(\mathbb E|G|=1/\sinh\eps\) and
\[
 \frac{\mathbb E|G|}{s}
 =\frac{1}{\tanh(\eps/2)\sinh\eps}
 =\frac{1}{2\sinh^2(\eps/2)}.
\]
Since $d=\mathrm{HD}(x,y)\in[0,n]$ and
$\widehat d=\operatorname{clip}_{[0,n]}(z)$,
projection onto the interval $[0,n]$ cannot increase the distance
to $d$. Hence
\[
|\widehat d-d|
\le
|z-d|.
\]
Thus, clipping cannot increase distance from \(d\).  The triangle inequality
and Cauchy--Schwarz give the stated error bound.  After clipping there
are only \(O(n)\) possible interior values plus two endpoints, so the
common output uses \(O(\log n)\) bits.
\end{proof}

\begin{algorithm}
\caption{Common-Output Retention-Robust Randomized-Response Protocol}
\label{alg:retention-robust-baseline}
\begin{algorithmic}[1]
\Require Alice has \(x\in\{0,1\}^n\); Bob has
\(y\in\{0,1\}^n\); privacy parameter \(\varepsilon>0\).
\Ensure A common output satisfying pure
\(\varepsilon\)-retention-robust QDP.

\State Set
\[
p\gets\frac{1}{1+e^\varepsilon},
\qquad
s\gets 1-2p=\tanh(\varepsilon/2).
\]

\For{\(i=1,\ldots,n\)}
    \State Alice independently samples
    \[
    N_i\sim\operatorname{Bernoulli}(p).
    \]
    \State Alice sets
    \[
    X_i'\gets x_i\oplus N_i.
    \]
\EndFor

\State Let
\[
X'=(X_1',\ldots,X_n')\in\{0,1\}^n.
\]

\State \textbf{Alice \(\rightarrow\) Bob:}
Send the classical message
\[
M_{A\rightarrow B}=X'.
\]
\Comment{Equivalently, send the basis state \(\ket{X'}\).}

\State If the message was sent as qubits, honest Bob measures it
in the computational basis.

\State Bob computes
\[
C
\gets
\sum_{i=1}^n
\mathbf{1}\{X_i'\neq y_i\}.
\]

\State Bob forms the debiased estimator
\[
\widetilde d
\gets
\frac{C-np}{s}.
\]

\State Bob independently samples a two-sided geometric random
variable \(G\in\Z\) with
\[
\Pr[G=g]
=
\frac{1-e^{-\varepsilon}}
     {1+e^{-\varepsilon}}
e^{-\varepsilon|g|}
\qquad (g\in\Z).
\]

\State Bob computes
\[
\widehat d
\gets
\clip_{[0,n]}
\left(
\widetilde d+\frac{G}{s}
\right).
\]

\State Bob encodes the result using a public finite encoding:
\[
W\gets\Enc_\varepsilon(\widehat d)
\in\{0,1\}^{O(\log n)}.
\]

\State \textbf{Bob \(\rightarrow\) Alice:}
Send the classical message
\[
M_{B\rightarrow A}=W.
\]

\State Alice decodes
\[
\widehat d_A
\gets
\Enc_\varepsilon^{-1}(W).
\]

\State Bob sets
\[
\widehat d_B\gets\widehat d.
\]

\State \Return
\[
\widehat d_A=\widehat d_B.
\]
\end{algorithmic}
\end{algorithm}

\paragraph{Technical Explanation of
Algorithm~\ref{alg:retention-robust-baseline}}
Algorithm~\ref{alg:retention-robust-baseline} is designed for the stronger
retention-robust model, in which a corrupted party may apply arbitrary CPTP
maps, retain auxiliary systems, change its prescribed messages, or abort.
Unlike the KHNP construction, its security does not rely on forcing a party
to return a non-orthogonal quantum state. Instead, every input-dependent
message is differentially private before it reaches the recipient.

Alice first applies independent randomized response to every input bit. With
\[
p=\frac{1}{1+e^\varepsilon},
\]
she sends
\[
X_i'
=
\begin{cases}
x_i, & \text{with probability }1-p,\\
1-x_i, & \text{with probability }p.
\end{cases}
\]
For adjacent inputs \(x\) and \(x'\), only one coordinate distribution
changes, and the likelihood ratio of the complete vector \(X'\) is at most
\(e^\varepsilon\). Therefore \(X'\) is already pure
\(\varepsilon\)-DP for Alice's input. Any operation subsequently performed by
a corrupted Bob (including coherent processing, retention, measurement, or
early termination) is postprocessing of this private message and cannot
weaken its privacy guarantee.

In an honest execution, Bob computes
\[
C=\sum_{i=1}^n\mathbf{1}\{X_i'\neq y_i\}.
\]
Let
\[
d=\HD(x,y),
\qquad
s=1-2p=\tanh(\varepsilon/2).
\]
Each coordinate that originally agrees contributes a mismatch with
probability \(p\), whereas each coordinate that originally disagrees
contributes a mismatch with probability \(1-p\). Hence
\[
\mathbb{E}[C\mid x,y]
=
np+s d,
\]
so
\[
\widetilde d=\frac{C-np}{s}
\]
is an unbiased estimate of \(d\).

The first message does not protect Bob's input from Alice, since Alice
generated it without using \(y\). Bob therefore adds an independent
two-sided geometric variable \(G\) before sending the common estimate back:
\[
\widehat d
=
\clip_{[0,n]}
\left(
\frac{C-np+G}{s}
\right).
\]
Conditional on every possible received value \(X'\), changing one bit of
\(y\) changes \(C\) by at most one. The geometric distribution satisfies the
corresponding \(e^\varepsilon\) likelihood-ratio bound, so \(C+G\), and hence
\(\widehat d\), is pure \(\varepsilon\)-DP for Bob's input.

This privacy argument remains valid even when Alice deviates from the
protocol. Honest Bob first measures Alice's purported message in the
computational basis, obtaining a classical string \(X'\). Conditional on that
outcome, Alice's retained quantum register is independent of \(y\). The
pointwise geometric-mechanism inequality can therefore be multiplied by the
same conditional retained state and summed over \(X'\), yielding the required
operator inequality for Alice's complete final view.

The two uses of \(\varepsilon\) do not compose to \(2\varepsilon\). Privacy is
required separately in the two directions. To protect \(x\) from Bob, the
only \(x\)-dependent transmission is Alice's randomized-response vector. To
protect \(y\) from Alice, Alice's initial message is independent of \(y\), and
only Bob's geometric release consumes privacy. Each protected input therefore
incurs a single \(\varepsilon\)-privacy loss.

For utility, independence across coordinates gives
\[
\operatorname{Var}(\widetilde d)
=
\frac{np(1-p)}{s^2}
=
\frac{n}{4\sinh^2(\varepsilon/2)}.
\]
Thus,
\[
\mathbb{E}
\left[
\left|\widetilde d-d\right|
\right]
\leq
\sqrt{\operatorname{Var}(\widetilde d)}
=
\frac{\sqrt n}{2\sinh(\varepsilon/2)}.
\]
Moreover,
\[
\frac{\mathbb{E}|G|}{s}
=
\frac{1}{2\sinh^2(\varepsilon/2)}.
\]
Since clipping to \([0,n]\) cannot increase the distance from
\(d\in[0,n]\), the triangle inequality yields
\[
\mathbb{E}
\left[
\left|\widehat d-d\right|
\right]
\leq
\frac{\sqrt n}{2\sinh(\varepsilon/2)}
+
\frac{1}{2\sinh^2(\varepsilon/2)}.
\]

The protocol uses \(n+O(\log n)\) classical bits, or equivalently
computational-basis qubits. It establishes a fully retention-robust baseline
but does not give a quantum advantage: its leading error remains
\(O(\sqrt n)\) for constant \(\varepsilon\). Whether genuinely quantum
communication can improve this accuracy while remaining secure against all
retention-capable CPTP strategies is left open.

\subsection{Statement on AI Use}

The authors generated the main ideas, algorithms, and protocols.
The pure DP protocol is also an approximate DP protocol
(for $\delta = 0$). However, using GPT 5.6 Sol,
the authors optimized the noise mechanism to obtain a more
accurate approximate DP protocol.
Furthermore, 
the images were generated in LaTeX via the use of GPT 5.6 Sol.
The authors independently checked the
mathematical expressions and figures.

\end{document}